\documentclass[12pt]{amsart}
\usepackage{url}
\usepackage{amsmath,amsxtra,amssymb,latexsym,epsfig,amscd,amsthm,fancybox,epsfig,bm}
\usepackage[mathscr]{eucal}
\usepackage{graphicx}
\usepackage{multicol,xcolor}
\usepackage{epsfig}
\usepackage{epstopdf}
\usepackage{cases}
\usepackage{subfig}
\usepackage{color}
\usepackage{hyperref}
\usepackage{bigints}
\usepackage{enumitem}
\usepackage{mathtools}

\usepackage{natbib}
\usepackage{hyperref}
\hypersetup{
    colorlinks=true,
    linkcolor=blue,
    filecolor=magenta,
    urlcolor=cyan,
    citecolor=blue,
}

\usepackage{float}
\newcommand{\Se}{\mathcal{S}}

\newtheorem{thm}{Theorem}[section]
\newtheorem {asp}{Assumption}[section]

\newtheorem{rmk}{Remark}[section]

\newtheorem{cor}{Corollary}[section]

\numberwithin{equation}{section}

\DeclareMathOperator{\Conv}{Conv}
\newcommand{\eps}{\varepsilon}

\newcommand{\M}{\mathcal{M}}

\newcommand{\E}{\mathbb{E}}
\newcommand{\BE}{\mathbf{E}}
\newcommand{\BB}{\mathbf{B}}

\newcommand{\BX}{\mathbf{X}}
\newcommand{\bx}{\mathbf{x}}

\newcommand{\by}{\mathbf{y}}

\newcommand{\bc}{\mathbf{c}}

\newcommand{\N}{\mathbb{N}}
\newcommand{\CN}{\mathcal{N}}
\newcommand{\PP}{\mathbb{P}}

\newcommand{\R}{\mathbb{R}}

\newcommand{\Lom}{\mathcal{L}}

\numberwithin{equation}{section}

\newcommand{\bed}{\begin{displaymath}}
\newcommand{\eed}{\end{displaymath}}
\newcommand{\bea}{\bed\begin{array}{rl}}
\newcommand{\eea}{\end{array}\eed}

\newcommand{\barray}{\begin{array}{ll}}
\newcommand{\earray}{\end{array}}
\newcommand{\diag}{{\rm diag}}
\def\disp{\displaystyle}

\newcommand{\1}{\boldsymbol{1}}

\newcommand{\dist}{\mathrm{dist}}

\def\bar{\overline}
\def\hat{\widehat}
\def\a.s{\text{\;a.s.\;}}
\def\supp{\text{supp\,}}
\def\dist{{\rm dist}}

\title[Stationary distributions]{Stationary distributions of persistent ecological systems}
\author[A. Hening]{Alexandru Hening* }
\thanks{*Corresponding Author}
\address{Department of Mathematics\\
Texas A\&M University\\
Mailstop 3368\\
College Station, TX 77843-3368\\
United States
}
\address{Department of Mathematics\\
Tufts University\\
Bromfield-Pearson Hall\\
503 Boston Avenue\\
Medford, MA 02155\\
United States
}
\email{al.hening@gmail.com}

\author[Y. Li]{Yao Li}
\address{Department of Mathematics and Statistics\\
University of Massachusetts Amherst\\
710 N Pleasant Street\\
Amherst, MA, 01003\\
United States
}
\email{yaoli@math.umass.edu}

\begin{document}
\maketitle
\begin{abstract}
We analyze ecological systems that are influenced by random environmental fluctuations. We first provide general conditions which ensure that the species coexist and the system converges to a unique invariant probability measure (stationary distribution). Since it is usually impossible to characterize this invariant probability measure analytically, we develop a powerful method for numerically approximating invariant probability measures. This allows us to shed light upon how the various parameters of the ecosystem impact the stationary distribution.

We analyze different types of environmental fluctuations. At first we study ecosystems modeled by stochastic differential equations. In the second setting we look at piecewise deterministic Markov processes. These are processes where one follows a system of differential equations for a random time, after which the environmental state changes, and one follows a different set of differential equations -- this procedure then gets repeated indefinitely. Finally, we look at stochastic differential equations with switching, which take into account both the white noise fluctuations and the random environmental switches.

As applications of our theoretical and numerical analysis, we look at competitive Lotka--Volterra, Beddington-DeAngelis predator-prey, and rock-paper-scissors dynamics. We  highlight new biological insights by analyzing the stationary distributions of the ecosystems and by seeing how various types of environmental fluctuations influence the long term fate of populations.
\end{abstract}
\tableofcontents

\section{Introduction}\label{s:intro}

One of the fundamental questions in population biology is understanding under what conditions interacting species coexist. It is well documented that one has to look carefully at the interplay between biotic interactions and environmental fluctuations when trying to determine criteria for coexistence or extinction. Sometimes biotic effects can result in species going extinct. However, if one adds the effects of the environment, extinction might be reversed into coexistence. These phenomena have been seen in competitive settings as well as in settings where prey share common predators \citep{CW81, AHR98, H77}. In other instances, deterministic systems that coexist become extinct once one takes into account environmental fluctuations \citep{HGT94}.
One successful way of analyzing the interplay between biotic interactions and environmental noise is by modelling the populations as discrete or continuous-time Markov processes. The problem of coexistence or extinction then becomes equivalent to studying the asymptotic behaviour of these Markov processes.

Throughout the years, the way ecologists think of environmental stochasticity has changed from seeing it as something that obfuscates the real patterns, and is therefore a nuisance, to viewing it as a mechanism that can create many new interesting behaviours \citep{BC18}. There are many different ways of modeling random environmental fluctuations. It is quite common for continuous time dynamics to focus on white noise fluctuations and to model the dynamics by a system of stochastic differential equations (SDE). For some systems, the randomness might not be best modelled by SDE \citep{T77} -- different situations will require different types of environmental stochasticity.

A natural way of analyzing the coexistence of species is looking at the average per-capita growth rate of a population when rare. Intuitively, if this growth rate is positive, the respective population increases when rare and can invade, while if it is negative the population decreases and goes extinct. If there are two populations, coexistence is ensured as long as each population can invade when it is rare and the other population is stationary \citep{T77, CE89, EHS15}. This criterion breaks down when one has more than two interacting populations \citep{SBA11}.

There is a general theory of permanence for deterministic models \citep{H81, H84, HJ89, HS98, ST11}. It can be shown that a sufficient condition for persistence is the existence of a fixed set of weights associated with the interacting populations, such that this weighted combination of the populations's invasion rates is positive for any invariant measure supported by the boundary (i.e. associated to a sub-collection of populations) -- see \cite{H81}. This coexistence theory has been generalized to stochastic difference equations \citep{SBA11, BS19, CHN19}, stochastic differential equations \citep{SBA11,HN16, CHN19}, and recently to general Markov processes \citep{B18}. The theory is not as well developed for ecosystems that involve both a continuous and a discrete component or ecosystems systems exhibiting multiple timescales. We close this gap by providing a number of new persistence results.

One natural class of processes that model the joint dynamics of the environment and the species densities is that of piecewise deterministic Markov processes (PDMP). The basic intuition behind PDMP is that due to different environmental conditions the functional way species interact can change. \cite{TL16} show that the predation behavior can vary with the environmental conditions and therefore change predator-prey cycles. Another example is that of plankton ecosystems where changing environmental conditions can facilitate coexistence \citep{H61,LC16,LK01}.
Since the environment is random, its changes (or switches) cannot be predicted in a deterministic way. For a PDMP, the process follows a deterministic system of differential equations for a random time, after which the environment changes, and the process switches to a different set of ordinary differential equations (ODE), and follows the dynamics given by this ODE for a random time. The procedure then gets repeated indefinitely.

A second class of processes we analyze is the one given by stochastic differential equations with switching (SSDE). These processes involve a discrete component that keeps track of the environment and which changes at random times. In a fixed environmental state the system is modelled by stochastic differential equations. This way we can capture the more realistic behaviour of two types of environmental fluctuations:
\begin{itemize}
  \item major changes of the environment (daily or seasonal changes),
  \item fluctuations within each environment.
\end{itemize}

We provide an in depth analysis in the deterministic/SDE/PDMP/SSDE settings for three ecological models. The first model is the classical Lotka--Volterra competition system with two species. The second model is a predator-prey system with Beddington-DeAngelis functional response. The third model is a system with rock-paper-scissors dynamics. For these three examples we look at the difference between the ODE, SDE, PDMP and SSDE frameworks.

Even though there are now powerful analytical results to study when species coexist, almost nothing is known about the stationary distribution of coexisting species. We make progress in this direction by developing a rigorous method of approximating the invariant probability measures of SDE, PDMP and SSDE. We prove rigorously that the approximations converge to the correct invariant probability measure. These new approximation methods are significantly more accurate than the usual Monte Carlo results. We make use of them in order to shed light on what the equilibrium distribution of species looks like.

The paper is organized as follows. In Sections \ref{s:SDE}, \ref{s:PDMP} and \ref{s:SSDE} we present persistence results for the frameworks of SDE, PDMP and SSDE. We describe our numerical methods for computing invariant probability measures in Section \ref{s:num_approx}. Section \ref{s:LV} provides an in depth analysis of a two-species Lotka--Volterra competition system. A predator-prey system with Beddington-DeAngelis functional response is explored in Section \ref{s:BD}. Results for three species competing according to rock-paper-scissors dynamics appear in Section \ref{s:RPS}. Finally, Section \ref{s:disc} is concerned with a discussion of our results.

\section{Stochastic differential equations} \label{s:SDE}
Consider the general nonlinear systems of the form
\begin{equation}\label{e:system}
dX_i(t)=X_i(t) f_i(\BX(t))dt+X_i(t)g_i(\BX(t))dE_i(t), ~i=1,\dots,n
\end{equation}
where
$\BE(t)=(E_1(t),\dots, E_n(t))^T=\Gamma^\top\BB(t)$
$\Gamma$ is a $n\times n$ matrix such that
$\Gamma^\top\Gamma=\Sigma=(\sigma_{ij})_{n\times n}$
and $\BB(t)=(B_1(t),\dots, B_n(t))$ is a vector of independent standard Brownian motions. The system \eqref{e:system} describes the dynamics of $n$ interacting populations whose densities at time $t\geq 0$ are given by $\BX(t):=(X_1(t), \dots, X_n(t))$.
We will denote by $\PP_\by(\cdot)=\PP(~\cdot~|~\BX(0)=\by)$ and $\E_\by[\cdot]=\E[~\cdot~|~\BX(0)=\by]$ the probability and expected value given that the process starts at $\BX(0)=\by$.

The drift term of our system $X_i(t)
f_i(\BX(t))$ is due to the deterministic dynamics (abiotic and biotic factors)
while the diffusion term $X_i(t)g_i(\BX(t))dE_i(t)$ is due to the effects of random environmental fluctuations. If the population densities at some point in time are $\bx$ then $f_i(\bx)$ is the per capita growth rate of the $i$th population/species if there are no environmental fluctuations. The matrix $\Gamma^T\Gamma$ captures the covariance structure of the environmental fluctuations.
\begin{rmk}
We note that just adding a stochastic fluctuating term to a
deterministic model has some short comings because it does not usually
give a mechanism for how different species are influenced by the environment. Instead, following the fundamental work by \cite{T77} we
see the SDE models as ``approximations for more realistic, but often
analytically intractable, models''. In particular, SDE can be seen as
scaling limits, or approximations, of stochastic difference
equations.
\end{rmk}

  The random environmental fluctuations make it impossible, under natural assumptions, for \eqref{e:system} to have non-trivial fixed points. Whenever $f_i(\bx_0)=0$, if $g_i(\bx_0)\neq 0$ the diffusion term will push $\BX$ away from $\bx_0$. This is one reason why fixed points are not usually useful concepts when studying ecosystems that are exposed to environmental fluctuations.
  If $\BX(0)=\bx\in \R^{n,\circ}_+:=(0,\infty)^n $ we say the population $X_i$ goes \textit{extinct} if for all $\bx\in\R^{n,\circ}_+$
\[
\PP_\bx\left\{\lim_{t\to\infty}X_i(t)=0\right\}=1.
\]
  A natural definition of persistence in a stochastic environment is the one given by \cite{C82}: the species $X_1,\dots, X_n$ \textit{persist in probability} if
for any $\eps>0$, there exists $\delta>0$ such that for any $\BX(0)=\by \in \R^{n,\circ}_+$ we have with probability 1 that
\[
\liminf_{t\to\infty}\PP_\by\left\{X_j(t)>\delta, j=1,\dots,n\right\} \geq 1-\eps.
\]
This definition says that with high probability the species stay away from the extinction set $\partial \R_+^{n,\circ}:= \R_+^n\setminus \R_+^{n,\circ}$. For any $\eta>0$ let
\[
\Se_\eta:=\{\bx,\in\R_+^{n,\circ}~:~\min_i x_i\leq \eta\}
\]
be the subset of $\R_+^{n,\circ}$ where at least one species is within $\eta$ of extinction.
For any $t\in \N$ define the \textit{normalized occupation measure}
\[
\Pi_t(B):=\frac{1}{t}\int_{0}^t\1\{\BX(s)\in(B)\}\,ds
\]
where $\1$ is the indicator function and $B$ is any Borel subset of $\R_+^{n,\circ}$. We note that $\Pi_t$ is a random probability measure and $\Pi_t(B)$ tells us the proportion of time the system spends in $B$ up to time $t$. We say $\BX(t)$ is \textit{almost surely stochastically persistent} \citep{BS19} if for all $\eps>0$ there exists $\eta(\eps)=\eta>0$ such that for all $\bx\in \R_+^{n,\circ}$
\[
\liminf_{t\to\infty} \Pi_t(\R_+^{n,\circ}\setminus \Se_\eta)>1-\eps, ~\BX(0)=\bx.
\]

In this paper we are concerned with a stronger version of persistence, which does not only tell us that the species stay away from extinction, but that they also converge to some type of random equilibrium.

The probability measure $\pi$ is an \textit{invariant probability measure} if, whenever one starts the process with initial conditions distributed according to $\pi$, then for any time $t\geq 0$ the distribution of $\BX(t)$ is given by $\pi$. In other words, $\pi$ is the random equivalent of a fixed point, or more generally, an attractor. A community of species which has an invariant probability measure $\pi$ has, loosely speaking, a random equilibrium characterized by $\pi$.

The process $\BX$ is said to be \textit{strongly stochastically persistent} if it has a unique invariant probability measure $\pi^*$ on $\R^{n,\circ}_+$ and
\begin{equation}
\lim\limits_{t\to\infty} \|P_\BX(t, \mathbf{x}, \cdot)-\pi^*(\cdot)\|_{\text{TV}}=0, \;\mathbf{x}\in\R^{n,\circ}_+
\end{equation}
where $\|\cdot,\cdot\|_{\text{TV}}$ is the total variation norm and $P_\BX(t, \mathbf{x}, A)=\PP(\BX(t)\in A~|~\BX(0)=\bx)$ is the transition probability of $(\BX(t))_{t\geq 0}$. We note that strong stochastic persistence implies persistence in probability and almost sure stochastic persistence \citep{HN16, BS19}.

Building a theory of stochastic persistence therefore boils down to finding conditions which imply the existence of a unique invariant probability measure $\pi^*$ which lives on the persistence set $\R^{n,\circ}_+$. We need a few more concepts in order to present this theory.
Let $\Conv\M$ denote the set of invariant measures of $\BX(t)$ whose support is contained in $\partial\R^n_+$. The set of extreme points of $\Conv\M$, denoted by $\M$, is the set of ergodic invariant measures with support on the boundary $\partial\R^n_+$. If a measure lives on the boundary $\partial\R^n_+$ it describes a strict subcommunity of species (at least one of the $n$ species is absent/extinct) that are at a random equilibrium.

If $\mu\in\M$ is an invariant measure and $\BX$ spends a lot of time close to its support, $\supp(\mu)$, then it will get attracted or repelled in the $i$th direction according to the \textit{Lyapunov exponent} (or \textit{invasion rate})
\begin{equation}\label{e:lyapunov_sde}
\lambda_i(\mu)=\int_{\partial\R^n_+}\left(f_i(\bx)-\frac{\sigma_{ii} g_i^2(\bx)}{2}\right)\mu(d\bx).
\end{equation}
\textit{\textbf{Biological interpretation:} Intuitively, $\lambda_i(\mu)$ tells us what happens if we introduce species $i$ at a very low density into the subcommunity whose equilibrium is described by $\mu$. If $\lambda_i(\mu)>0$ species $i$ tends to persist while if $\lambda_i(\mu)<0$ species $i$ tends to go extinct. The quantity $\lambda_i(\mu)$ averages the log growth rate of the dynamics of species $i$ around the random equilibrium given by $\mu$.
}

The following theorem \citep{SBA11, HN16, B18, CHN19} gives a powerful criterion for persistence.

\begin{thm}\label{t:p_sde}
If for any $\mu\in\Conv(\M)$,
$$\max_{i=1,\dots,n}\lambda_i(\mu)>0,$$
and the technical assumptions from Appendix \ref{s:a_SDE} hold, then the system is strongly stochastically persistent: there exists a unique invariant measure $\mu^*$ with support on $\R^{n,\circ}_+$, and $\BX$ converges to $\mu^*$ in total variation
$$
\lim\limits_{t\to\infty} \left\|P(t,\mathbf{x}, \cdot)-\mu^*(\cdot)\right\|_{TV}=0, \bx\in\R^{n,\circ}_+.
$$
Furthermore, the rate of convergence is exponential.
\end{thm}
\textit{\textbf{Biological Interpretation:} If each subcommunity of species that persists, and is characterized by an invariant probability measure $\mu\in \Conv(\M)$, can be invaded by at least one species that is not part of the subcommunity, then the full community of $n$ species persists. When the process gets close to the extinction boundary, the species which are close to extinction will grow/decay exponentially fast according to their invasion rates. Since at least one invasion rate is positive, the process gets pushed away from the boundary and make extinction impossible.
}
\subsection{Single species ecosystems} Suppose the species is governed by
\[
dX(t) = X(t) f(X(t))dt + g(X(t))dB(t)
\]
Then, if the Lyapunov exponent is positive, i.e.
\[
f(0) - \frac{g^2(0)}{2}>0
\]
by Theorem \ref{t:p_sde} the species persists and converges to a unique invariant probability measure $\mu^*$.

\section{Piecewise deterministic Markov processes}\label{s:PDMP}

As before, we look at a system of $n$ unstructured interacting populations, let $X_i(t)$ denote the density of the $i$th population at time $t\geq 0$, and set $\BX(t)=(X_1(t),\dots,X_n(t))$. Suppose $(r(t))$ is a continuous-time process taking values in the finite state space $\CN=\{1,\dots,n_0\}.$ This Markov chain keeps track of the environment, so if $r(t)=i\in\CN$ this means that at time $t$ the dynamics takes place in environment $i$. Once one knows in which environment the system is, the dynamics are given by a system of ODE. The PDMP can therefore be written

\begin{equation}\label{e1-pdm}
dX_i(t)=X_i(t)f_i(\BX(t),r(t))dt, i=1, \dots, n,
\end{equation}
where $f_i$ now depends on both the species densities and the environmental state $r(t)$.
In order to have a well-defined system one has to specify the switching-mechanism, e.g. the dynamics of the process $(r(t))$.
We assume that the switching intensity of $r(t)$ depends on the state of $\BX(t)$ as follows
\begin{equation}\label{e:tran}\begin{array}{ll}
&\disp \PP\{r(t+\Delta)=j~|~r(t)=i, \BX(s),r(s), s\leq t\}=q_{ij}(\BX(t))\Delta+o(\Delta) \text{ if } i\ne j \
\hbox{ and }\\
&\disp \PP\{r(t+\Delta)=i~|~r(t)=i, \BX(s),r(s), s\leq t\}=1+q_{ii}(\BX(t))\Delta+o(\Delta).
\end{array}
\end{equation}
where $q_{ii}(\bx):=-\sum_{j\ne i}q_{ij}(\bx)$. We assume that $q_{ij}(\bx)$ is a bounded continuous function for each $i,j\in\CN$ and the matrix $Q(\bx)=(q_{ij}(\bx))_{n_0\times n_0}$ is irreducible. It is well-known that a process $(\BX(t),r(t))$ satisfying \eqref{e1-pdm} and \eqref{e:tran}
is a strong Markov process \citep{D84}.

The simplest case would be to assume that the transition matrix $Q(\bx)=(q_{ij}(\bx))_{n_0\times n_0} = (q_{ij})_{n_0\times n_0}$ is independent of the densities of the species. In this case the jump times will be exponentially distributed and independent of the process $\BX$. Nevertheless, it is useful from an ecological standpoint to allow the state of the environment $r(t)$ to be influenced by the population $\BX$. This is because there can be feedbacks between the environment and the species living in it. One example would be ecosystem engineers like beavers or oysters -- these change the structure of the environment and thereby create a population-environment feedback \citep{JLS94, CWH09, MLPS16}. Another example would be a fire-vegetation feedback \citep{SL12}.

Piecewise deterministic Markov processes have been used recently by \cite{BL16, HN20} to prove how competitive exclusion can be reversed into coexistence. The authors look at a two dimensional competitive Lotka-Volterra system in a fluctuating environment. They show that the random switching between two environments that are both favorable to the same species, e.g. the favored species persists and the unfavored species goes extinct, can lead to the extinction of this favored species and the persistence of the unfavored species, to the coexistence of the two competing species, or to bistability, where depending on the initial condition one species persists and the other goes extinct.

In \cite{HS19} the authors look at a system that switches between two deterministic classical Lotka-Volterra predator-prey systems -- the assumption is that there are no intraspecific competition terms for the prey or predator species. Even though for each deterministic predator-prey system the predator and the prey densities form closed periodic orbits, it is shown that the switching makes the system leave any compact set. Moreover, in the switched system, the predator and prey densities oscillate between $0$ and $\infty$.

The above examples show how PDMP can exhibit dynamics that are radically different from those of each fixed environment.
\subsection{Mathematical framework}
The quantity $\PP_{\bx,r}(A)$ will denote the probability of event $A$ if $(\BX(0), r(0))=(\bx,k)$. Call $\mu$ an invariant measure for the process $\BX$ if $\mu(\cdot,\cdot)$ is a measure such that for any $k\in \CN$ $\mu(\cdot,k)$ is a Borel probability measure on $\R_+^n$ and, if one starts the process with initial conditions distributed according to $\mu(\cdot,\cdot)$, then for any time $t\geq 0$ the distribution of $(\BX(t), r(t))$ is given by $\mu(\cdot,\cdot)$.

Let $\Conv\M$ denote the set of invariant measures of $(\BX(t),r(t))$ whose support is contained in $\partial\R^n_+\times\CN$. The set of extreme points of $\Conv\M$, denoted by $\M$, is the set of ergodic invariant measures with support on the boundary $\partial\R^n_+\times\CN$.

If $\mu\in\M$ is an invariant measure and $\BX$ spends a lot of time close to its support, $\supp(\mu)$, then it will get attracted or repelled in the $i$th direction according to the \textit{Lyapunov exponent} (or \textit{invasion rate})
\begin{equation}\label{e:lyapunov_pdmp}
\lambda_i(\mu)=\sum_{k\in\CN}\int_{\partial\R^n_+}f_i(\bx, k)\mu(d\bx,k).
\end{equation}
\textit{\textbf{Biological Interpretation:} The invasion rate $\lambda_i(\mu)$ is the average per-capita growth rate of species $i$ when it is introduced at a low density into the subcommunity of species characterized by the invariant probability measure $\mu$. The averaging is done by weighing the different growth rates $f_i(\bx, k)$ according to the measure $\mu$. Note that $f_i(\bx, k)$ is the growth rate when the densities of the $n$ species are $\bx$ and the environment is in state $k$. According to equation \eqref{e:lyapunov_pdmp} the averaging has to be done over species densities as well as the different environments.
}

 The following theorem tells us when there the system exhibits coexistence -- see \cite{BL16,B18, HS19, HN20} for proofs in particular settings.

\begin{thm}\label{t:p_pdmp}
If for any $\mu\in\Conv(\M)$,
$$\max_{i=1,\dots,n}\lambda_i(\mu)>0,$$
i.e. $\mu$ is a repeller, then the system is strongly stochastically persistent, i.e. 
for any $\epsilon>0$, there exists a $\delta>0$ such that
\begin{equation}
\liminf\limits_{t\to\infty} \PP_{\bx,r}\{X_i(t)\geq\delta, i=1,\dots,n\}\geq1-\eps, \bx\in\R_+^{n,\circ}, r\in \CN.
\end{equation}
Under additional irreducibility conditions (see Appendix \ref{s:a_PDMP}), there exists a unique invariant measure $\mu^*$ with support on $\R^{n,\circ}_+\times\CN$, and $\BX$ converges to $\mu^*$ in the following sense
$$
\lim\limits_{t\to\infty} \left\|P(t,(\mathbf{x}, r), \cdot)-\mu^*(\cdot)\right\|_{TV}=0, (\bx,r)\in\R^{n,\circ}_+\times\CN
$$
where $\|\cdot,\cdot\|_{\text{TV}}$ is the total variation norm and $P(t,(\mathbf{x}, r),\cdot)$ is the transition probability of $(\BX(t), r(t))_{t\geq 0}$. Furthermore, the rate of convergence is exponential.
\end{thm}

\subsection{Single species systems} Consider a single species whose dynamics is given by
\[
\frac{dX}{dt}(t)= X(t) f(X(t),r(t)).
\]
Since the matrix $q(0)$ is irreducible, the associated Markov chain has a stationary distribution $\nu:=(\nu_1,\dots,\nu_{n_0})$ - this is the Markov chain we get if we start with $X(0)=0$.
Note that the extinction set is $\{0\}\times \CN$. The only invariant probability measure supported on this set is $\mu_0:=\delta_0\times \nu$.
As a result of Theorem \ref{t:p_pdmp}, if
\[
\lambda(\mu_0) = \sum_{k=1}^{n_0}\nu_kf(0,k)>0
\]
and some irreducibility conditions hold (see Appendix \ref{s:a_PDMP} or \cite{BBMZ15, B18}) then there is a unique invariant probability measure $\mu^*$ and $X(t)$ converges in total variation to $\mu^*$.

\textit{\textbf{Biological Interpretation:} If the species persists in both environments then it will persist with the environmental switching. If the species does not persist in any environment, then it will also not persist with the environmental switching. However, it is possible to have source and sink environments and still have persistence in the fluctuating system. Suppose there are only two environments, and that $f(0,1)<0, f(0,2)>0$. Then, $\{0\}$ is an attractor for $f(x,1)$ and the species cannot persist in environment $1$. Similarly, $\{0\}$ is not an attractor for $f(x,2)$ and the species persists in environment $2$. However, as long as $\nu_1 f(0,1) + \nu_2 f(0,2)>0$, or equivalently
\[
\frac{\nu_1}{\nu_2}< \left|\frac{f(0,2)}{f(0,1)}\right|
\]
persistence is possible. If the fraction of time spent in the sink environment over the fraction of the time spent in the source environment is smaller than the fraction of the per capita growth rates at $0$ in the source and sink environments, then the system with switching is persistent.
}
\section{Stochastic differential equations with switching}\label{s:SSDE}
In this section we analyze ecological systems that can be modelled by stochastic differential equations with switching (SSDE). These processes are similar to PDMP. The difference is that now, instead of switching between systems of ODE, one switches between systems of SDE.
One gets
\begin{equation}\label{e:SSDE}
dX_i(t)=X_i(t)f_i(\BX(t),r(t))dt + X_i g_i(\BX(t),r(t))dE_i, i=1,\dots,n
\end{equation}
where $\BE(t)=(E_1(t),\dots, E_n(t))^T=\Gamma^\top\BB(t)$ for an $n\times n$ matrix
$\Gamma$ such that
$\Gamma^\top\Gamma=\Sigma=(\sigma_{ij})_{n\times n}$, $\BB(t)=(B_1(t),\dots, B_n(t))$ is a vector of independent standard Brownian motions and. The process $r(t)$ keeps track of the environment state and lives in $\CN=\{1,\dots,n_0\}.$ The switching intensity of the discrete process $r(t)$ will be modelled by
\begin{equation}\label{e:tran2}\begin{array}{ll}
&\disp \PP\{r(t+\Delta)=j~|~r(t)=i, \BX(s),r(s), s\leq t\}=q_{ij}(\BX(t))\Delta+o(\Delta) \text{ if } i\ne j \
\hbox{ and }\\
&\disp \PP\{r(t+\Delta)=i~|~r(t)=i, \BX(s),r(s), s\leq t\}=1+q_{ii}(\BX(t))\Delta+o(\Delta).
\end{array}
\end{equation}
One assumes, as in the PDMP setting, that $q(\cdot)$ is a bounded function which depends continuously on $\bx$. Furthermore, we assume that $q(\bx)$ is irreducible for all $\bx\in \R_+^n$. General properties for these processes have been studied thoroughly \citep{YZ09, ZY09, NYZ17}. However, there are few results regarding the persistence of ecological systems modelled by SSDE. We present below the framework and a powerful persistence result.

\subsection{Mathematical framework}
The quantity $\PP_{\bx,r}(A)$ will denote the probability of event $A$ if $(\BX(0), r(0))=(\bx,k)$. Call $\mu$ an invariant measure for the process $\BX$ if $\mu(\cdot,\cdot)$ is a measure such that for any $k\in \CN$ $\mu(\cdot,k)$ is a Borel probability measure on $\R_+^n$ and, if one starts the process with initial conditions distributed according to $\mu(\cdot,\cdot)$, then for any time $t\geq 0$ the distribution of $(\BX(t), r(t))$ is given by $\mu(\cdot,\cdot)$.

Let $\Conv\M$ denote the set of invariant measures of $(\BX(t),r(t))$ whose support is contained in $\partial\R^n_+\times\CN$. The set of extreme points of $\Conv\M$, denoted by $\M$, is the set of ergodic invariant measures with support on the boundary $\partial\R^n_+\times\CN$.

If $\mu\in\M$ is an invariant measure and $\BX$ spends a lot of time close to its support, $\supp(\mu)$, then it will get attracted or repelled in the $i$th direction according to the \textit{Lyapunov exponent} (or \textit{invasion rate})
\begin{equation}\label{e:lyapunov_ssde}
\lambda_i(\mu)=\sum_{k\in\CN}\int_{\partial\R^n_+}\left(f_i(\bx, k)-\frac{\sigma_{ii}g_i^2(\bx,k)}{2}\right)\mu(d\bx,k).
\end{equation}
We note that if we let $g_i=0$ in \eqref{e:lyapunov_ssde} we get the expression \eqref{e:lyapunov_pdmp} that holds for PDMP. If we have only one switching state, $n_0=1$, \eqref{e:lyapunov_ssde} to the SDE setting from \eqref{e:lyapunov_sde}. In this sense

 Using the methods of \cite{B18} or generalizing \cite{HN16} one can prove the following persistence result.

\begin{thm}\label{t:p_ssde}
If for any $\mu\in\Conv(\M)$,
$$\max_{i=1,\dots,n}\lambda_i(\mu)>0,$$
then the system is strongly stochastically persistent, i.e. 
for any $\eps>0$, there exists a $\delta>0$ such that
\begin{equation}
\liminf\limits_{t\to\infty} \PP_{\bx,r}\{X_i(t)\geq\delta, i=1,\dots,n\}\geq1-\eps, \bx\in\R_+^{n,\circ}, r\in \CN.
\end{equation}
Under additional irreducibility conditions, there exists a unique invariant measure $\mu^*$ with support on $\R^{n,\circ}_+\times\CN$, and $\BX$ converges to $\mu^*$ in the total variation
$$
\lim\limits_{t\to\infty} \left\|P(t,(\mathbf{x}, r), \cdot)-\mu^*(\cdot)\right\|_{TV}=0, (\bx,r)\in\R^{n,\circ}_+\times\CN.
$$
Furthermore, the rate of convergence is exponential.
\end{thm}

\subsection{Single species systems} Consider a single species whose dynamics is given by
\[
dX(t)= X(t) f(X(t),r(t))dt + X(t)g(X(t),r(t))dB(t).
\]
As argued in the PDMP example, the matrix $q(0)$ is irreducible, so the associated Markov chain has a stationary distribution $\nu^0:=(\nu_1^0,\dots,\nu_{n_0}^0)$.
Once again, the only invariant probability measure living one the extinction set (or boundary) is $\mu_0:=\delta_0\times \nu^0$.
As a result of Theorem \ref{t:p_ssde}, if
\begin{equation}\label{e:mu_0}
\lambda(\mu_0) = \sum_{k=1}^{n_0}\nu_k^0\left(f(0,k)-\frac{g^2(0,k)}{2}\right)>0
\end{equation}
and some weak irreducibility conditions hold \citep{B18} then the species persists, there is a unique invariant probability measure $\mu^*$, and $X(t)$ converges in total variation to $\mu^*$.

\textit{\textbf{Biological Interpretation:} The stochastic per-capita growth rate at $0$ of the switching system, $\lambda(\mu_0)$ is a weighted combination of the stochastic per-capita growth rates at $0$ of the system in each fixed environment, $f(0,k)-\frac{g^2(0,k)}{2}$. The weights are given by the fractions of time the process $r(t)$ spends in each environmental state $k\in \{1,\dots,n_0\}$ when one starts the system at extinction ($X(0)=0$). Say the species goes extinct for $k\in M\subsetneq \{1,\dots,n_0\}$
\[
f(0,k)-\frac{g^2(0,k)}{2}<0 , k\in M.
\]
It is still possible to have persistence when
\[
\sum_{k\notin M }\nu_k^0f(0,k)-\frac{g^2(0,k)}{2} > \sum_{k\in M}\nu_k^0 \left|f(0,k)-\frac{g^2(0,k)}{2}\right|.
\]
Even if some environments are sinks the species can still persist as long as it spends enough time in the source environments.
}

\section{Numerical computation of invariant probability measures}\label{s:num_approx}

The dynamics of populations are random due to the inherent stochastic
nature of environmental fluctuations. As discussed in the above sections,
ecological systems can be modeled by SDE, SSDE, or PDMP. Recent developments have made it
possible to get robust conditions for the persistence and extinction of
these models \citep{HN16, B18, BS19, CHN19}. In certain cases one can show that the
populations converge to a unique invariant probability measure or
steady state. If the ecosystem consists of a single species, e.g. the stochastic dynamics is one-dimensional, there are well known ways in which one can describe the invariant probability measure of a persistent system. If the
dimension is greater or equal to two, i.e. the system has at least two
species or there is more than one environment (for PDMP and SSDE), one can almost never describe the invariant probability
measure analytically. Characteristics of the steady state distribution
could provide valuable information regarding the distribution of
species to theoretical ecologists. We propose to use a new method
introduced in \cite{li2019data} to numerically approximate invariant probability
measures of ecological systems in an efficient fashion.

\subsection{Invariant probability measures of SDE}

For simplicity, we use the following SDE to illustrate the numerical method:
\begin{equation}\label{e:sde1}
dX_i(t)=X_if_i(\BX(t))\,dt + X_i(t) g_i(\BX(t))\,dB_i(t), i=1,\dots,n.
\end{equation}
where $f_i, g_i:\R_+^n\to\R_+$ are continuous  and $(B_1,\dots,B_n)$
is an $n$-dimensional Brownian motion. Without loss of generality, we
assume that $\mathbf{X}$ admits a unique solution and is strongly
stochastically persistent. This implies that $\mathbf{X}$ has a
unique invariant probability measure on $(0, \infty)^{n}$. Let $\bar f_i(\bx):= x_i f_i(\bx)$ and $\bar g_i(\bx):= x_i g_i(\bx).$

Under natural assumptions the distribution of $\BX(t)$ has a probability density function $\ell(t,\bx)$ such that for any measurable set $A$
\[
\PP(\BX(t)\in A) = \int_A \ell(t,\by)\,d\by.
\]
The time evolution of $\ell(t,\bx)$ is described by the \textit{Fokker-Planck equation}
\begin{equation}\label{e:FP}
\begin{split}
\ell_t &= \mathcal{L} \ell (\bx) = -\sum_{i=1}^n \partial_{x_i}(\bar f_i \ell) + \frac{1}{2}\sum_{i,j=1}^n \partial_{x_ix_j}(D_{i,j}\ell)\\
\ell(0,\bx)&=\ell_0(\bx)
\end{split}
\end{equation}
where $D=(\bar g(\bx))^T\Gamma^T\Gamma\bar g(\bx)$ and $\ell_0(\bx)$ is the probability density function
of $\BX(0)$. A invariant probability density function $\ell_*$ is a
probability density function that satisfies the \textit{stationary
  Fokker-Planck equation}
\begin{equation}\label{e:FPS}
   \mathcal{L}\ell_*=0.
\end{equation}
Note that if a solution $\ell_*$ of \eqref{e:FPS} exists, and satisfies $\ell_*\geq 0$ and $\int _{\R_+^{n}}\ell^*(\bx)\,d\bx$ then it defines an invariant probability measure $\pi$.

There are usually two approaches for numerically finding $\ell_*$. The
first one is to solve the Fokker-Planck equation numerically. The main
problem is the boundary condition of the PDE. Since $(0, \infty)^{n}$ is an
unbounded domain, one will have to set $(0, L)^{n}$ as the
numerical domain for a sufficiently large $L$, give zero boundary condition
at $\{x_{i} = L\}$ and a reflecting boundary condition at $\{x_{i} =
0\}$. This method will cause many problems. First, the resultant
linear system only has an identically zero solution. One needs to add a constraint
that the integral of the solution is $1$, and find the least
square solution instead. Second, the diffusion term in most
ecological models vanishes at $\{x_{i} = 0\}$. This singular boundary
condition can cause more error if a reflecting boundary condition is
used. Third, since $L$ has to be very large in practice, this can easily
make the computational cost too high.

A second approach is
using Monte Carlo simulations. For this method one collects $N$
samples of $\BX(t)$, and counts the number of samples in each bin of a grid to
estimate the probability density function. In order to get a high accuracy, one needs to use
a very large number $N$ of samples. This approach also suffers from
the curse of dimensionality: if the dimension $n$ is large then one
needs $N$ to be unrealistically larger.

The solution is to use a data-driven method to compute the invariant
probability measure \cite{li2019data}. This approach works for stochastic difference equations, stochastic differential
equations, stochastic differential equations with switching, and
piecewise deterministic Markov processes. The idea is to use a
combination of the PDE and Monte Carlo (MC) methods. One first
generates a reference solution using MC. This reference solution is then used as a replacement of
the boundary conditions from the PDE method. This new method will
combine the high accuracy of the PDE method and the flexibility of MC
simulations. We explain how to do this for equation \eqref{e:sde1} in a two-dimensional domain
$\mathcal{D}=[a_0,b_0]\times [a_1,b_1]\subset \R_+^{2,\circ}$.

 Construct an $N\times M$ grid with grid size $h=(b_0-a_0)/N =
 (b_1-a_1)/M$. Approximate $\ell$ at the center of each of the
 $N\times M$ boxes $O_{i,j}=[a_0+(i-1)h,a_0+ih]\times
 [a_1+(j-1)h,a_1+jh]$. Denote by ${\bm \ell}^{N,M}$ the numerical
 solution on $\mathcal{D}$ - note that ${\bm \ell}$ can be seen as
 a vector in $\R_+^{NM}$. An entry $\ell_{i,j}$ of ${\bm \ell}$
 will be an approximation of the density function $\ell$ at the center
 of the box $O_{i,j}$. If we discretize the Fokker-Planck equation
 with respect to the centers of the boxes we will get a linear
 constraint $$A{\bm \ell}=0$$ where $A$ is a matrix. The matrix $A$ is the
 so-called \textit{discretized Fokker-Planck operator}. If $D$ is diagonal, a row of equation
 $A{\bm \ell}=0$ reads
\begin{align*}
 &- \frac{1}{h}(\bar f_{1}^{i+1,j}\ell_{i +1, j} - \bar f_{1}^{i+1,j}\ell_{i+1,j})
 - \frac{1}{h}(\bar f_{2}^{i,j+1}\ell_{i, j+1} - \bar f_{2}^{i,j-1}\ell_{i,j-1})
 + \frac{1}{h^{2}}(D_{1,1}^{i+1,j} \ell_{i+1,i} \\
+&
 D_{1,1}^{i-1,j}\ell_{i-1,j}+D_{2,2}^{i,j+1}\ell_{i,j+1} +
 D_{2,2}^{i,j-1}\ell_{i,j-1} - 2 D_{1,1}^{i,j}\ell_{i,j} -
   2D_{2,2}^{i,j}\ell_{i,j})  = 0 \,,
\end{align*}
where $\bar f_{1}^{i,j}$ (resp. $\bar f_{2}, D_{11}, D_{22}$) is the value of
$\bar f_{1}$  (resp. $\bar f_{2}, D_{11}, D_{22}$) at the center of box
$O_{ij}$.

\textbf{Step 1:} Get a reference solution using MC. Define $(\tilde
\BX_n)_{n=1}^N$ for $n=1,\dots, T$ to be the numerical trajectory of
the chain $ \BX(n\delta)$ where $\delta>0$ is the time step of the MC
simulation. Define $\mathbf{v} = \{v_{i, j}\}_{i= 1, j = 1}^{i = M, j
  = N}$ with $v_{i,j} = \frac{1}{Th^2}\sum_{n=1}^T
\mathbf{1}_{O_{i,j}}(\tilde X_n)$. One can show that $\mathbf{v}$ is an
approximate solution to \eqref{e:sde1}.

\textbf{Step 2:} Solve the optimization problem:
\begin{align}
\label{optimization}
  \mbox{min} &\|{\bm \ell}-\mathbf{v}\|_2\\
\mbox{subject to} & A{\bm \ell} =0 \,.
\end{align}

It follows from Theorem \ref{convergence} that the optimization problem
will reduce the error in $\mathbf{v}$, because it projects the
error term from $\mathbb{R}^{NM}$ to $\mathrm{Ker}(A)$, whose
dimension is much smaller. As a result, the norm of the error term can
be significantly reduced by this projection.

\textit{ Intuitively the algorithm we propose works as follows: First get a
  low-accuracy numerical solution using Monte Carlo simulation. Once
  we have this reference solution, we solve an optimization problem,
  which looks for the least squares solution with respect to the
  reference solution, under the constraint given by the numerical
  discretization scheme.This method is more efficient because it does not need to look at
boundary conditions for PDE.}

\subsection{Invariant probability measures of SSDE and PDMP}

In the previous section we explained how our data-driven method works for SDE. We next explain how to extend this to PDMP and SSDE. Take the SSDE
\begin{equation}
  \label{SSDE}
dX_i(t)=X_i(t)f_i(\BX(t),r(t))dt + X_i g_i(\BX(t),r(t))dE_i, i=1,\dots,n
\end{equation}
introduced in Section \ref{s:SSDE} as an example, where $r(t)$ is an independent continuous time Markov chain on $\mathcal{N} = \{1, \cdots, n_{0}\}$. Let $\bar f_i = x_i f_i$ and $\bar g_i = x_ig_i$. One can see \citep{YZ09}
that the probability density function of the invariant probability
measure of equation \eqref{SSDE}, denoted by $\ell_{*}(k, \mathbf{x})$,
must satisfy
\begin{equation}
\label{switchFPE}
  0 = \mathcal{L} \ell_{*}(k, \mathbf{x}) = - \sum_{i =
    1}^{n}\partial_{x_{i}}(\bar f_{i}(k, \mathbf{x}) \ell_{*}(k,
  \mathbf{x})) + \frac{1}{2}\sum_{i,j =
    1}^{n} \partial_{x_{i}x_{j}}(D_{i,j}(k, \mathbf{x}) \ell_{*}(k,
  \mathbf{x})) + \sum_{i = 1}^{n_{0}} \ell_{*}(i,
  \mathbf{x})q_{ik}(\mathbf{x})  \,
\end{equation}
together with $\ell_*(k, \mathbf{x}) \ge 0$ and $\int_{\R_{+}^n} \ell_*(k, \mathbf{x})dx =1$ where $D=\bar g^T\Gamma^T\Gamma \bar g$.
When computing $\ell_{*}$ numerically, we can still divide the state
space $(0, \infty)^{n}\times \mathcal{N}$ into bins.

We look at a 2d
SSDE example. We have a numerical domain $\mathcal{D} = [a_{0},
b_{0}]\times [a_{1}, b_{1}] \times \mathcal{N}$. As in the previous
subsection, the $(i,j)$-th box in the grid is denoted by $O_{i.j}$. After constructing an
$N \times M$ grid on $\mathcal{D}$, a numerical solution ${\bm \ell}
:= {\bm \ell}^{MN}$ has the form ${\bm \ell} = \{\ell_{i,j,k}\}_{i = 1,
j = 1, k = 1}^{i = N, j = M, k = n_{0}}$, where $\ell_{i,j,k}$
represents the probability density function $\ell$ at the center of box
$O_{i,j}$ with state $k$.

The first step is still to run a long trajectory to get
the reference solution $\mathbf{v}$, as described in the previous
subsection. Then we generate a linear constraint $A {\bm \ell} = 0$,
where $A$ comes from a discretization of equation
\eqref{switchFPE}. The next step is to solve the constrained optimization problem
$$\mbox{min} \| {\bm \ell} - \mathbf{v} \|_{2}$$
subject to the constraint $$A {\bm \ell} = 0.$$
This gives a numerical invariant probability density function ${\bm
  \ell}$ which will be a good approximation to the solution of \eqref{switchFPE}.

The case of PDMP is more delicate. Although theoretically the data-driven
solver should still work, in practice the invariant probability measure
can be too singular for a numerical scheme to capture its
density. If for instance, the probability density function takes a
very high value in one grid, but zero value in its neighbor grid, then
some numerical artifacts can be observed. Our simulation shows that
when the invariant probability measure of the PDMP model is not too
singular, the data-driven solver can still successfully reduce the
error in the reference solution given by the Monte Carlo
simulation. On the other hand, if the invariant probability measure is very singular,
usually the solution produced by a Monte Carlo simulation is accurate enough. This happens because the MC simulation result usually gives a clear dichotomy: a small box
will either get enough sample points to keep a satisfactory accuracy,
or zero sample points which imply a zero probability density.

\subsection{Block data-driven solver}

For higher dimensional problems, it is important to divide the domain
into many subdomains to reduce the scale of the numerical linear
algebra problem. This is called the block data-driven solver. The motivation is that (i) the data-driven
Fokker-Planck solver does not rely on the boundary condition, and (ii)
the optimization problem ``pushes'' most error terms to the boundary
of the domain. (See our discussion in \cite{dobson2019efficient}.) As a result we can
solve a similar optimization problem on each subdomain in parallel,
and merge all ``local solutions'' together. Finally, we need to deal with
the interface error on the boundary of each subdomain. This is done by
(i) letting the subdomains overlap with each other by a few grids, and
(ii) redrawing subdomains such that the new subdomain covers the old
subdomain interfaces, and (iii) running the block data-driven solver again. This
approach works for SDE, PDMP, and SSDE. We
refer to  \cite{dobson2019efficient} for the full details.

\subsection{Analysis of the algorithm}
We give proofs that our algorithm converges in Appendix
\ref{s:a_conv}. According to Theorem \ref{convergence}, the numerical
error depends on both the grid size and the quality of the Monte Carlo
sampler. The empirical error is better than the theoretical prediction
because the error term concentrates at the domain boundary and can be
eliminated by other approaches, such as the block data-driven solver
introduced above.

The quality of the Monte Carlo sampler plays an important role in the
data-driven solver. Monte Carlo sampling from the invariant
probability measure $\mu^{*}$ usually means running a numerical
trajectory of the SDE (or PDMP, SSDE) for a long time. However, a
numerical trajectory is only an approximation of a true SDE
trajectory. The invariant probability measure of the Monte Carlo
sampler, denoted by $\hat{\mu}$, is usually different from
$\mu^{*}$. Hence an analysis of the sensitivity of $\mu^{*}$ against
the numerical approximation is extremely important. In this paper, we
adopt the sensitivity analysis algorithm proposed in
\cite{dobson2019using} to study the sensitivity of the invariant
probability measures in our examples. The idea is to use the extrapolation method
to estimate the finite time error. Then the rate of contraction of the
Markov kernel, which
can be numerically computed by using couplings, is used to
extend the estimate to infinite time. We refer readers to Appendix \ref{s:a_conv}
for the full details.

\section{Lotka--Volterra competitive dynamics}\label{s:LV}
\subsection{ODE}\label{s:LV_ODE}
The classical Lotka--Volterra competitive system has the form
\begin{equation}\label{e:ODE_2d_k}
\begin{array}{lll}\disp d X_1(t)=X_1(t)\left(\ell_1-a_{11} X_1(t) -a_{12} X_2(t)\right)dt ,
\\
dX_2(t) =X_2(t)\left(\ell_2 - a_{22}X_2(t)-a_{21}X_1(t)\right)\,dt.
\end{array}
\end{equation}
Here $\ell_i>0$ is the per-capita growth rate of species $i$ and $a_{ij}>0$ is the per-capita competition rate between species $i$ and $j$. It is well known that both species persist when
\[
\ell_i - a_{ij}\frac{\ell_j}{a_{jj}}>0, i\neq j, i\in\{1,2\}.
\]

\subsection{SDE}\label{s:LV_SDE}
Suppose there are environmental fluctuations. Then \eqref{e:ODE_2d_k} becomes
\begin{equation}\label{e:SDE_2d_k}
\begin{array}{lll}\disp d X_1(t)=X_1(t)\left(\ell_1-a_{11} X_1(t) -a_{12} X_2(t)\right)dt +\sigma_1 X_1(t)\,dB_1(t) ,
\\
dX_2(t) =X_2(t)\left(\ell_2 - a_{22}X_2(t)-a_{21}X_1(t)\right)\,dt+ \sigma_2 X_2(t)\,dB_2(t).
\end{array}
\end{equation}
The dynamics for this system is well-known \citep{KO81, EHS15, HN17b, HN17, HN16}.
Species $i$ persists when species $j\neq i$ is absent if
\[
\lambda_i(\delta_0) = \ell_i - \frac{\sigma_i^2}{2}>0.
\]
If the above condition holds species $i$ has a unique invariant measure when it evolves on its own, $\mu_i$.
The conditions for coexistence, in the fixed environment $k$, are that for all $i, j\in\{1,2\}$ with $i\neq j$ we have
\begin{equation}
\label{eq7-3}
  \lambda_i(\mu_j) = \ell_i -\frac{\sigma_i^2}{2} -a_{ij} \frac{\ell_j-\frac{\sigma_j^2}{2}}{a_{jj}} >0.
\end{equation}
\subsection{PDMP}\label{s:LV_PDMP} Suppose we have two competing species that interact in a habitat where the environment switches between two possible states. Then
\begin{equation}\label{e:PDMP_2d}
\begin{array}{lll}\disp d X_1(t)=X_1(t)\left(\ell_1(r(t))-a_{11}(r(t)) X_1(t) -a_{12}(r(t)) X_2(t)\right)dt ,
\\
dX_2(t) =X_2(t)\left(\ell_2(r(t)) - a_{22}(r(t))X_2(t)-a_{21}(r(t))X_1(t)\right)\,dt.
\end{array}
\end{equation}
Suppose $q_1, q_2$ are the switching rates and $(\nu_1,\nu_2)$ the stationary distribution of the Markov chain $r(t)$.
If
\[
\nu_1\ell_i(1)+\nu_2\ell_i(2)>0, i=1,2
\]
each species can persist on its own. Furthermore, if species $i$ is on its own it has a unique invariant measure $\mu_i$ on $(0,\infty)\times \{1,2\}$ \citep{B18}. If $p_i(1):=\frac{\ell_i(1)}{a_{11}(1)}\neq \frac{\ell_i(2)}{a_{11}(2)}:=p_i(2)$ the probability measure $\mu_i$ has a density with respect to Lebesgue measure, in the sense that
\[
\mu_i(dx,k)= h_i(x,k)\1_{[p_i(1),p_i(2)]}(x)\,dx
\]
where $\1_{[p_i(1),p_i(2)]}(x)$ is the indicator of the interval $[p_i(1),p_i(2)]$ and $h_i(x,k)$ can be found explicitly as a function of the model parameters.
The invasion rates can also be computed by
\[
\lambda_j(\mu_i) = \sum_k \int_0^\infty h_i(x,k)\1_{[p_i(1),p_i(2)]}(x)\,dx.
\]

\begin{thm}
Suppose both species survive on their own, i.e. $\nu_1\ell_i(1)+\nu_2\ell_i(2)>0, i=1,2$. If $\lambda_1(\mu_2)>0, \lambda_2(\mu_1)>0$ the two species persist and the system converges to a unique stationary distribution $\mu^*$ on $(0,\infty)^2\times\{1,2\}$.
\end{thm}
The work of \cite{BL16, MZ16, MP16, HN20} contains the full classification of the dynamics. \cite{BL16, HN20} show how competitive exclusion can become coexistence due to the fluctuations of the environment, even though in both environments species $1$ goes extinct and species $2$ persists.

\subsection{Numerical examples}

We numerically analyze the finite time trajectory and the invariant probability
measure for the SDE, PDMP, and SSDE versions of the Lotka-Volterra competition model.

1) For the
SDE model, we take two different sets of parameters: $l_{1} = 4, l_{2}
= 2, a_{11} = 2,
a_{12} = 2, a_{21} = 0.4, a_{22} = 1.2$
(Parameter $1$) and $l_{1} = 2, l_{2} = 4, a_{11} = 0.8, a_{12} = 1.6,
a_{21} = 1, a_{22} = 5$ (Parameter $2$). The magnitude of the white noise fluctuations is
$\sigma_{1} = \sigma_{2} = 1$. It
follows from Equation \eqref{eq7-3} that two species can co-exist for
both parameter sets, with a stable equilibrium $(0.5, 1.5)$ and $(1.5,
0.5)$ respectively. Figure
\ref{fig1LV} Top shows trajectories of equation \eqref{e:SDE_2d_k} with
both parameter sets. Figure \ref{fig1LV} Bottom shows two invariant probability density functions of \eqref{e:SDE_2d_k} on
$(0, 2)^{2}$ with respect to two parameter sets. The mesh size of our
numerical computation is $800 \times 800$. The reference solution is
obtained from a Monte Carlo simulation with $10^{8}$ sample
points. One can see that the invariant probability density functions
tend to concentrate near the stable equilibrium.

\textit{\textbf{Biological Interpretation:} The white noise environmental fluctuations turn the stable equilibrium of the two species into stationary distributions that concentrate around the deterministic equilibria. How spread the invariant measure is seems to be related to the intraspecific competition rates. As expected, the higher the intraspecific competition rate is for species $i$, the smaller the spread of the stationary distribution is in direction $i$ - see \ref{fig1LV}. High intraspecific competition makes it harder for the species to leave the area that is close to the deterministic stable equilibria.
}

2) We next consider the PDMP version of the Lotka-Volterra competition
model. The number of environments is taken to be $n_{0}=2$. The model
switches between two parameter sets introduced above. We consider the
invariant probability measure for PDMP with respect to three different
speed of switching, slow, medium, and fast. Rates of random switching
are $q_{12} = 2.5, q_{21} = 4$ for slow switching, $q_{12} = 5, q_{21}
= 8$ for medium switching, and $q_{12} = 10, q_{21} = 16$ for fast
switching. Three invariant probability measures are demonstrated in
Figure \ref{fig2LV}. We can see that when the switching is slow, two marginal
invariant probability measures with respect to two states are very
different. From Figure \ref{fig2LV} Top one can see how trajectories
travelling between two stable equilibria with respect to two parameter
sets. One interesting phenomenon is that, with slow switching rates, deterministic trajectories
have enough time to approach to the invariant manifold corresponding
to the largest eigenvalue. (See in particular Figure \ref{fig1LV} top
left panel.) This makes the invariant probability measure very singular. When the switching becomes faster, a deterministic trajectory
can only move a shorter time between switchings. As a result, two
marginal invariant probability measures are more similar to each
other, as seen in Figure \ref{fig2LV} Mid and Bottom. The numerical domain is still
$(0, 2)^{2}$ with a mesh size $600 \times 600$. The reference solution is obtained from a Monte Carlo
simulation with $10^{8}$ sample points.

\textit{\textbf{Biological Interpretation:} If the switching is slow the species spend a long time in each environment. Because of this they have time to get close the equilibria. This is in line with the explanation for the paradox of the plankton by \cite{H61}. As switching becomes faster, the dynamics spends less time in a fixed environment, and the species do not have time to get close to the stable equilibrium from that environment. The invariant probability function from \ref{fig2LV} Middle reflects this fact - it is less singular than when the switching is slow. However, as we increase the switching and we go from intermediate switching to fast switching, the stationary distribution once again becomes more singular. For very fast switching the dynamics approaches the one given by a mixed deterministic ODE and therefore the mass of the stationary distribution concentrates around the equilibrium of a mixed ODE.
}

3) The PDMP model that jumps between two systems without coexistence can
still admit an invariant probability measure supported in the first
quadrant \citep{HN20}. We exhibit such an example in Figure \ref{fig2p}. The parameters
are $a_{21} = 1$ in state $1$ and $a_{12} = 1.6$ in state $2$. All the
other parameters are the same as in the previous examples. One can check that coexistence is not possible in environment $1$ or $2$ - in one environment species $1$ persists and species $2$ goes extinct while the opposite happens in the other environment. We consider the
invariant probability measure for the PDMP for three different
switching rates - the same as in the example from Figure \ref{fig2LV}. We can see here a
similar phenomenon. The two marginal distributions in the two states get
closer to each other as the switching rate becomes faster. The numerical
domain is $(0.6, 1.6) \times (0, 1)$ with a mesh size $600 \times
600$. The reference solution is obtained from a Monte Carlo
simulation with $10^{9}$ sample points.

The parameters in Figure
\ref{fig2b} are $l_{1} = 1, l_{2} = 0.5, a_{11} = 1, a_{12} = 0.2,
a_{21} = 1$, and $a_{22} = 0.2$ at state $1$, $l_{1} = 7.8, l_{2} =
15.2, a_{11} = 2, a_{12} = 0.4, a_{21} = 4$, and $a_{22} =0.8$ at
state $2$. The rate of switching is $q_{12} = 1.4$ and $q_{21} =
5.0$. It is easy to check that species $1$ goes extinct in both
environments. However, the PDMP model has coexistence for the specific
switching rates we picked - see \cite{HN20} for a proof. The invariant probability measure of the PDMP model
is presented in Figure \ref{fig2b} Top. One can see that the invariant
probability measure concentrates heavily near the extinction set and
near a line segment in state $2$. The numerical
domain is $(0, 4) \times (0, 16)$ with a mesh size $500 \times
2000$. The invariant probability measure is obtained from a Monte Carlo
simulation with $10^{9}$ sample points. Data-driven PDE solver is not
available for this example because the invariant probability density
function is too singular.

\textit{\textbf{Biological Interpretation:} In contrast to Hutchinson's explanation \citep{H61} we see that coexistence is possible for a wide range of switching rates. For a slow switching regime the invariant distribution in the two environments is more singular - the species tend to move towards the equilibrium on the $x$ axis in one environment and towards the equilibrium on the $y$ axis in the other environment. For fast switching, even though the time to extinction is much longer, we get once again that, in contrast to the prediction by \cite{H61}, coexistence is possible. In this case the dynamics of the two species is close to the mixed dynamics from the two environments.
}

4) We look at the SSDE model. Parameters of the deterministic parts
are taken to be the same as before. The strength of the white noise fluctuations is $\sigma_{1} =
\sigma_{2} = 0.5$ in both environmental  states. Once again, we study the invariant
probability measure under three different switching rates. The rates of random switching
are picked to be $q_{12} = 0.5, q_{21} = 0.8$ for slow switching, $q_{12} = 2.5, q_{21}
= 4$ for medium switching, and $q_{12} = 10, q_{21} = 16$ for fast
switching. The numerical result is shown in Figure \ref{fig3LV}. A phenomenon similar to what happened in the previous example can be observed: the two marginal invariant
probability measures move towards each other when the speed of
switching becomes higher. The numerical domain is $(0, 2)^{2}$ with a mesh size $600 \times 600$. The reference solution is obtained from a Monte Carlo
simulation with $10^{8}$ sample points.

 We also compute the invariant
probability measure of the SSDE model corresponding to the PDMP model
from Figure \ref{fig2b}, in which the same species goes extinct in both
environments. All parameters are the same as in the PDMP model. The
magnitude of the noise is $\sigma_{1} = \sigma_{2} = 0.1$. In the presence
of noise, the concentration close to the extinction set is more
significant.

\textit{\textbf{Biological Interpretation:} In this setting we note that the invariant probability measure is more spread out - the white noise fluctuations diffuse the equilibrium distribution of the species. As the switching rates increase, the invariant probability measures in the two environmental states become more and more similar. Intuitively, the invariant probability measures converge to the invariant probability measure of an SDE.
}

The last step is the sensitivity analysis. We use the Milstein scheme to
simulate SDE and SSDE. Hence the estimator of the mean finite time
error is
$$
  I = \frac{1}{N}\sum_{i = 0}^{N - 1}\rho(\hat{X}^{dt}_{i,T},
  \hat{X}^{2dt}_{i,T}) \,.
$$
Parameters in our simulation are $dt = 0.001$, $N = 10^{6}$. For the SDE
model, we only run simulation using the first parameter set. The time
span is $T = 3$ for SDE, and $T = 2$ for SSDE. The estimator $I$ is
$0.00511$ for SDE, and $0.00497$ for SSDE. Then we run coupling method to show the speed
of convergence. The exponential tails of the coupling times are showcased in
Figure \ref{fig4}. We can see that the coupling time distribution decays exponentially in each case. This gives us an estimate of
$\alpha \approx 0.35$ for the SDE and $\alpha \approx 0.29$ for the SSDE. Therefore, we
have $\mathrm{d}_{w}( \mu^{*}, \hat{\mu}) \approx 0.00786$ for the SDE, and
$\mathrm{d}_{w}( \mu^{*}, \hat{\mu}) \approx 0.0070$ for the SSDE. We
conclude that the numerical results for the approximation of invariant probability
measures are reliable for our SDE and SSDE examples.

 \begin{figure}[!htbp]
 	\begin{center}
 		\includegraphics[width = \linewidth]{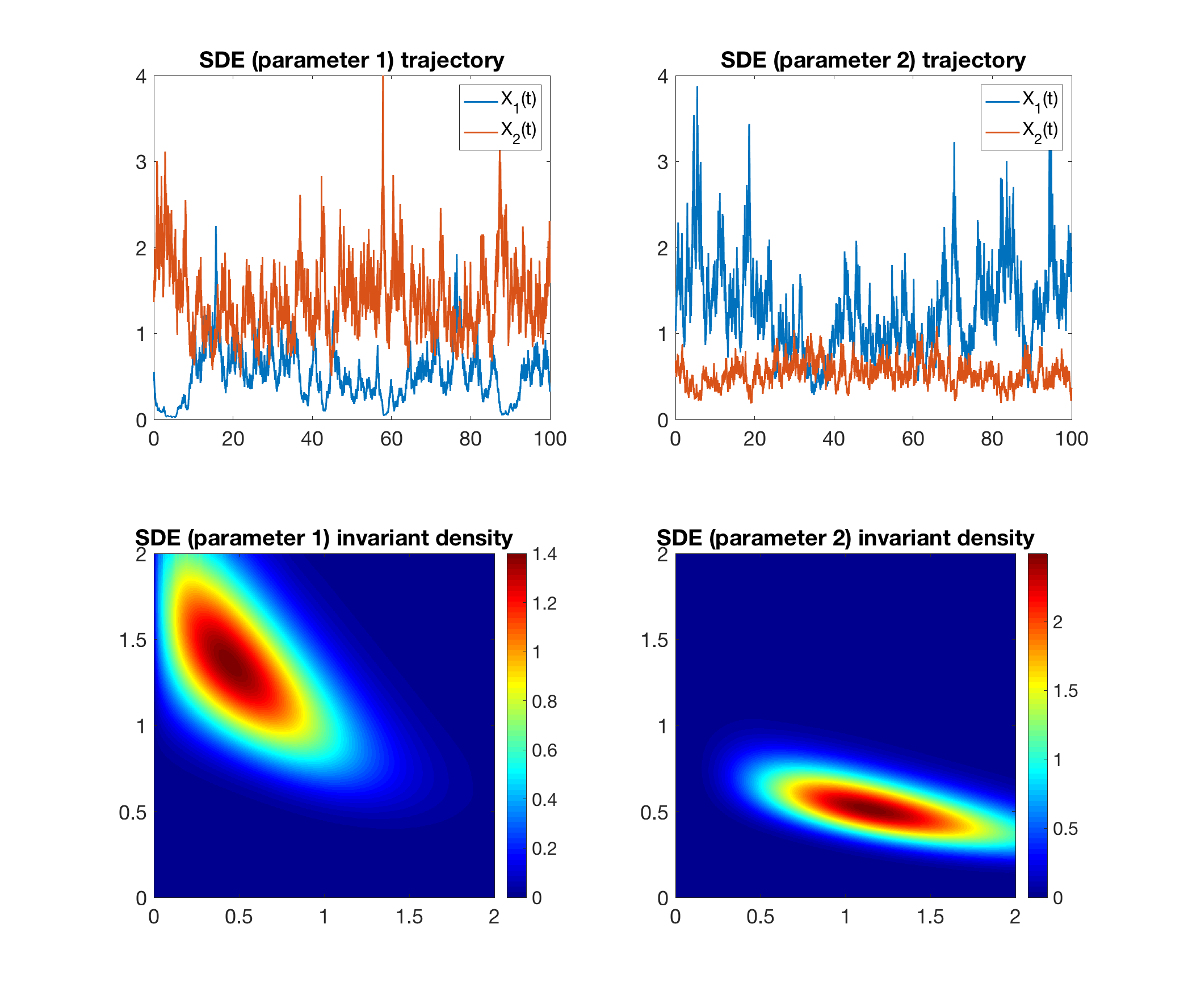}
 		\caption{Trajectories and invariant probability
                  measures of the Lotka-Volterra SDE model. {\bf Top}: Trajectories of $X_{1}(t)$ and
                  $X_{2}(t)$ for both parameter sets. {\bf Bottom}:
                  Invariant probability measures for both parameter
                  sets. }
\label{fig1LV}\end{center}
 \end{figure}

\begin{figure}[!htbp]
 	\begin{center}
 		\includegraphics[width = \linewidth]{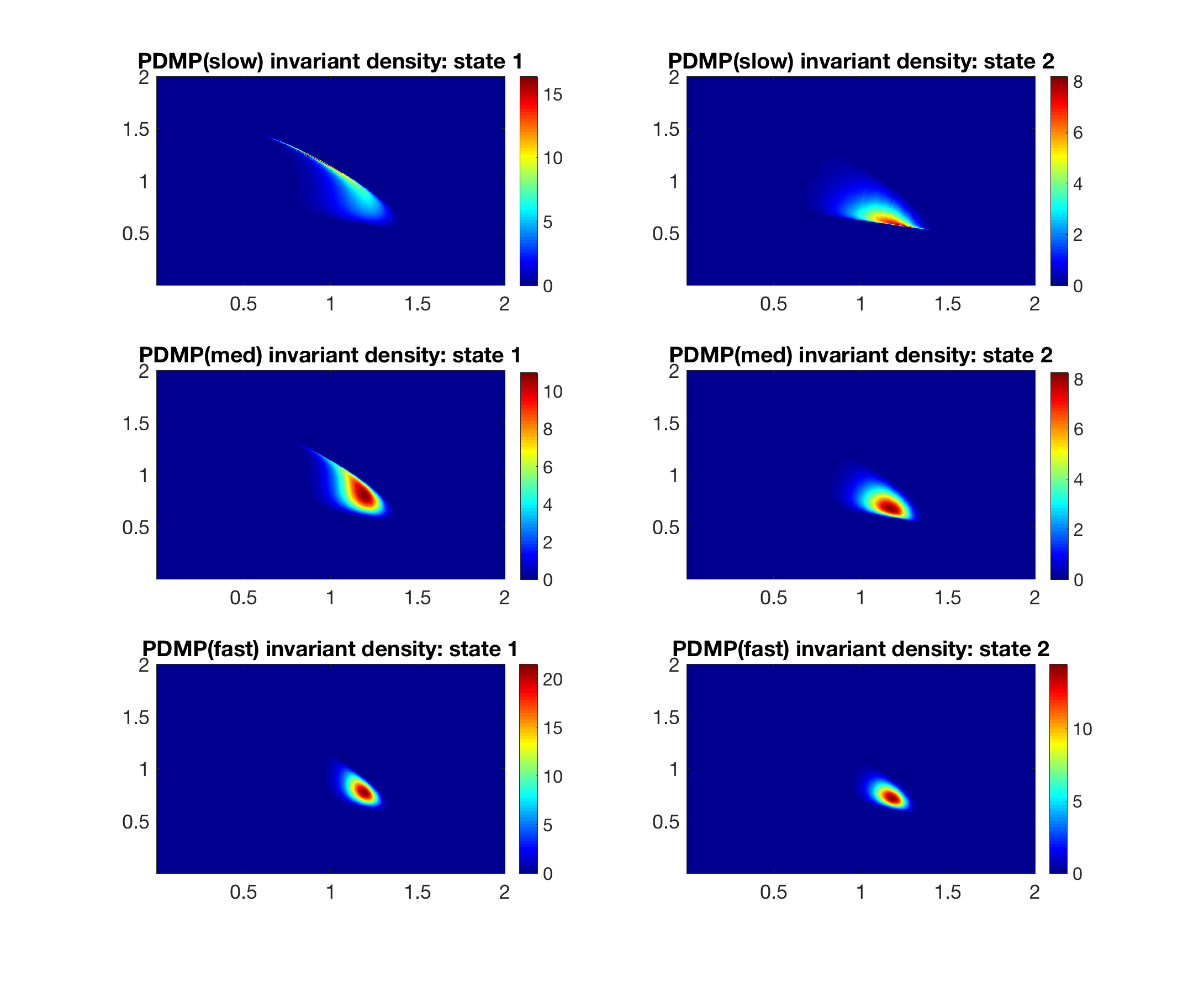}
 		\caption{Invariant probability measures of
                  Lotka-Volterra PDMP model. The deterministic system exhibits coexistence in each environmental state. Top to bottom: Invariant probability measures
                 for three different sets of switching rates. }
\label{fig2LV}\end{center}
 \end{figure}

\begin{figure}[!htbp]
 	\begin{center}
 		\includegraphics[width = \linewidth]{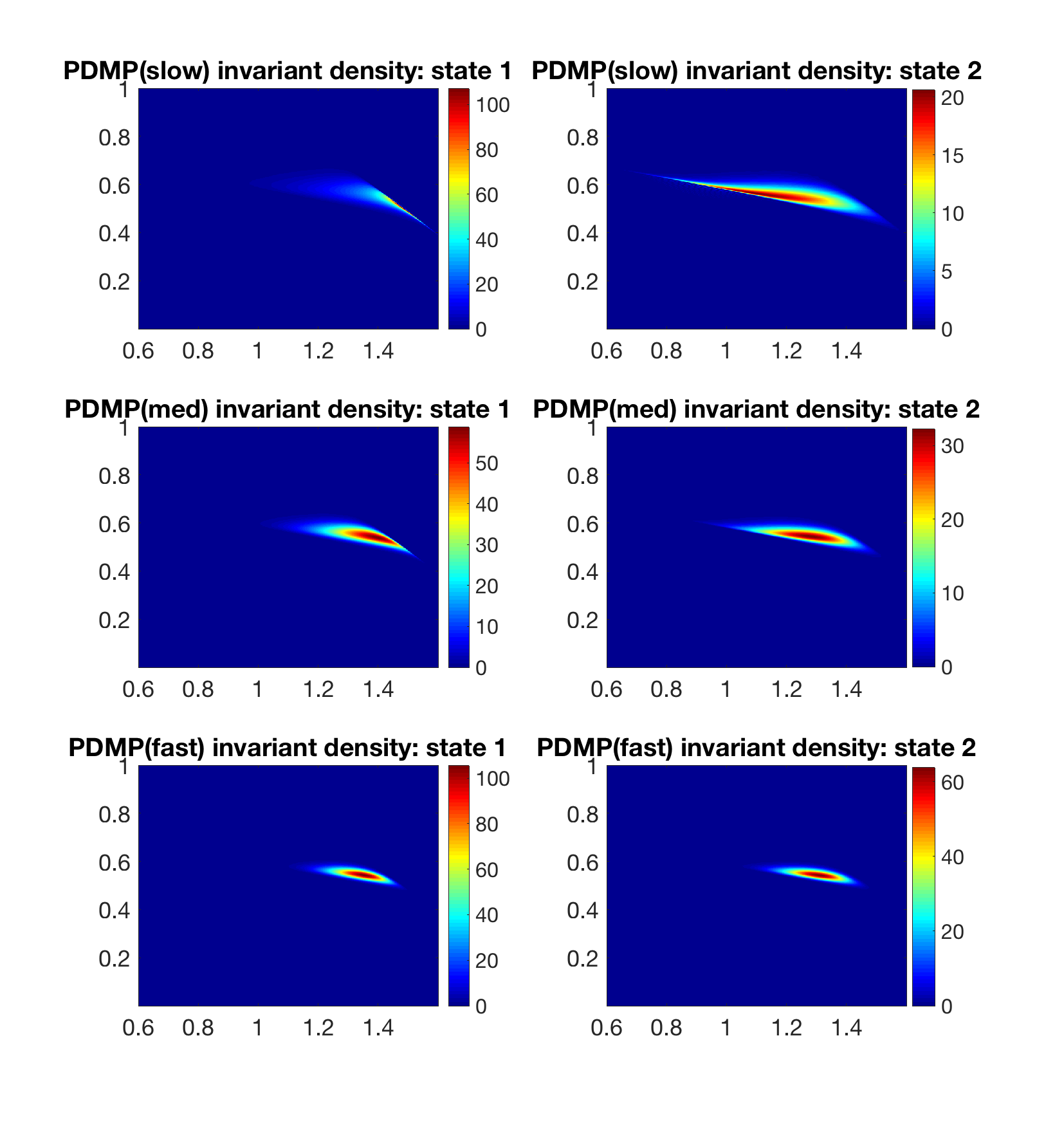}
 		\caption{Invariant probability measures for the
                  Lotka-Volterra PDMP model. The deterministic system does not exhibit coexistence in any of the two environmental states - one species persists in environment 1 and the other in environment 2. Top to bottom: Invariant probability measures
                 for three different sets of switching rates. }
\label{fig2p}\end{center}
 \end{figure}

\begin{figure}[!htbp]
 	\begin{center}
 		\includegraphics[width = \linewidth]{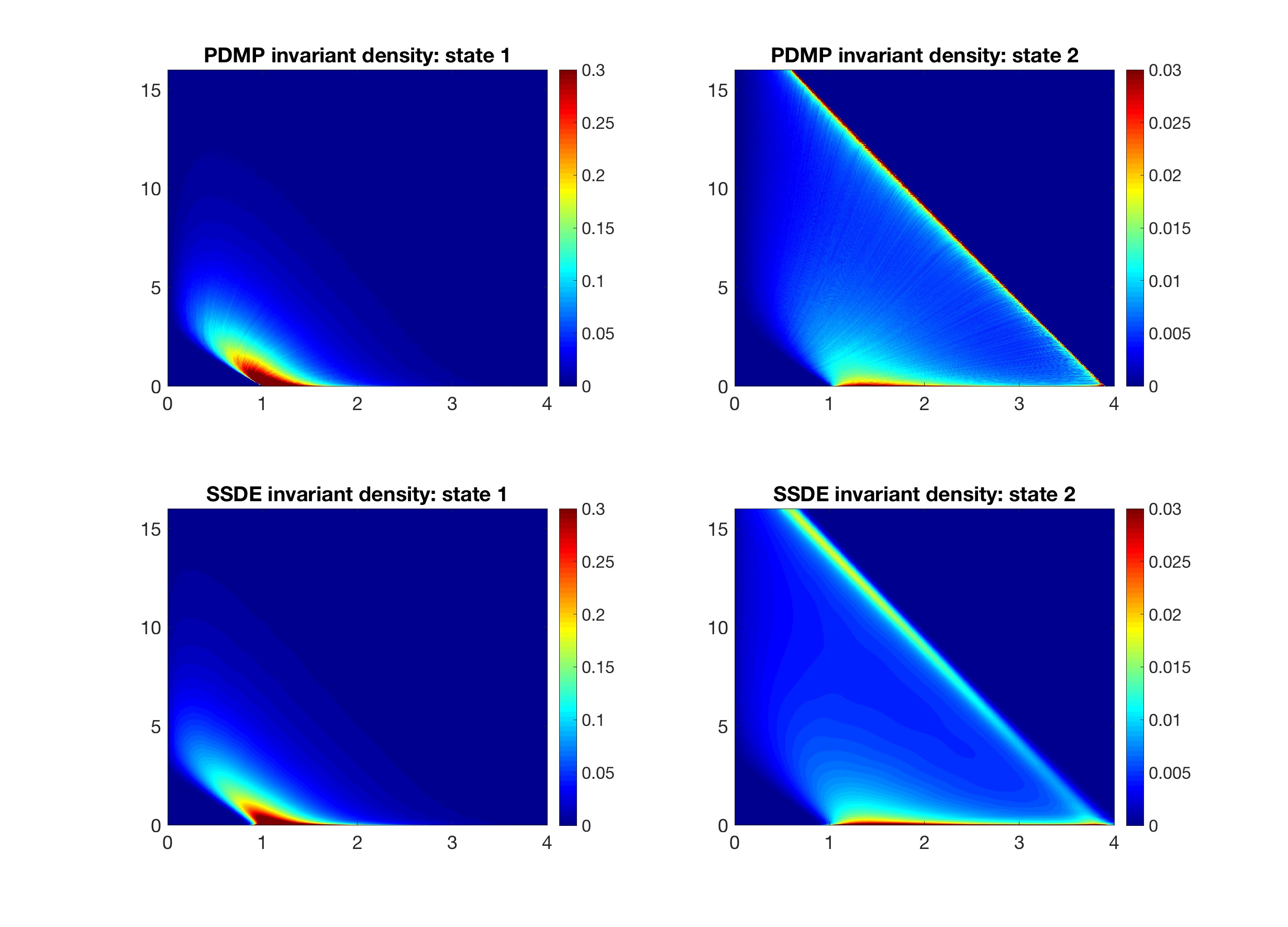}
 		\caption{Invariant probability measures of the
                  Lotka-Volterra PDMP model and SSDE model. Species
                  $1$ goes extinct in both deterministic environments. Top: Invariant probability measures
                 of the PDMP model at two states. Bottom:  Invariant probability measures
                 of the SSDE model at two states}
\label{fig2b}\end{center}
 \end{figure}

\begin{figure}[!htbp]
 	\begin{center}
 		\includegraphics[width = \linewidth]{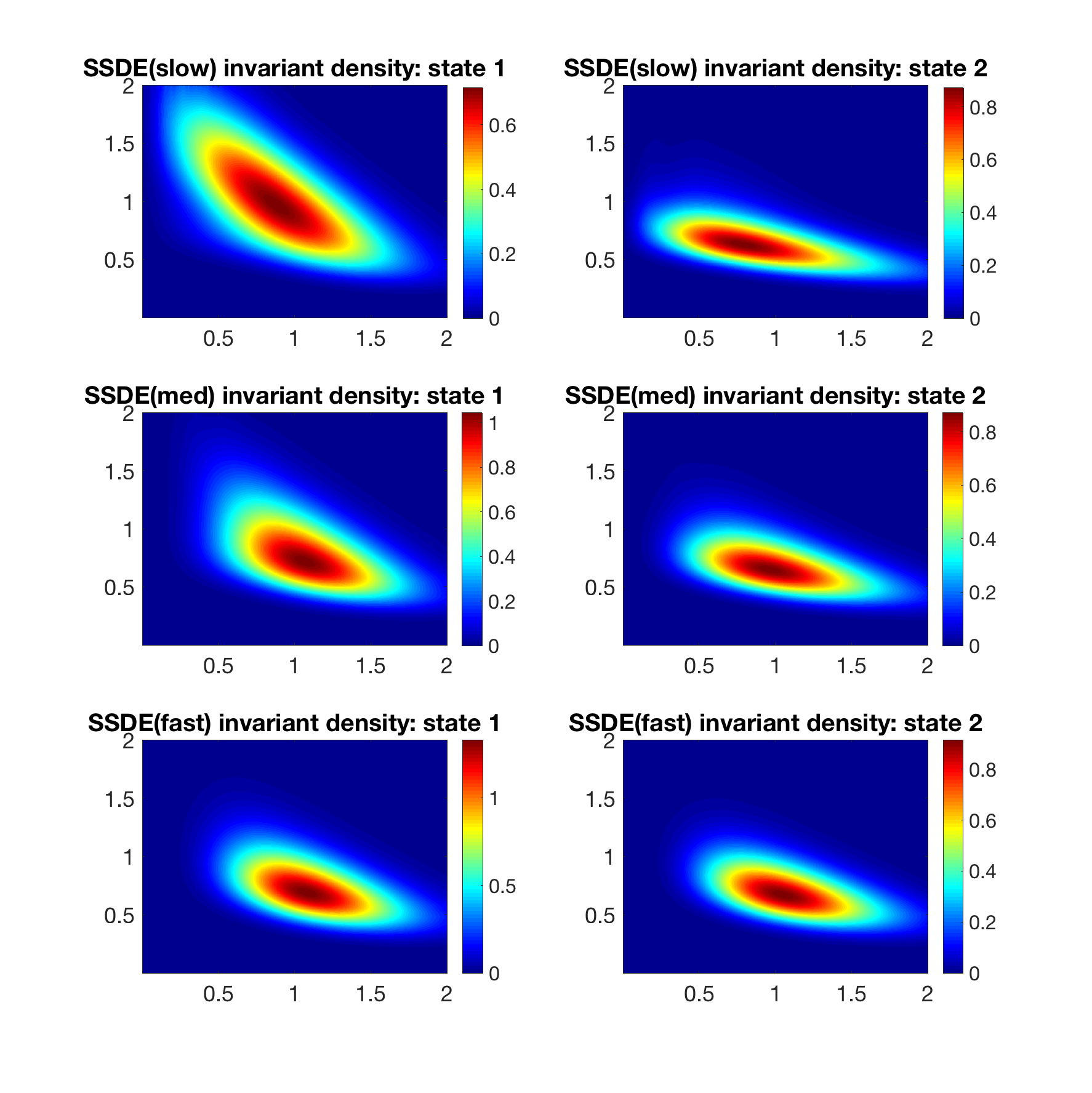}
 		\caption{Top to bottom: Invariant probability measures
                of the SSDE Lotka-Volterra model for three different sets of switching rates. }
\label{fig3LV}\end{center}
 \end{figure}

\begin{figure}[!htbp]
 	\begin{center}
 q 		\includegraphics[width = \linewidth]{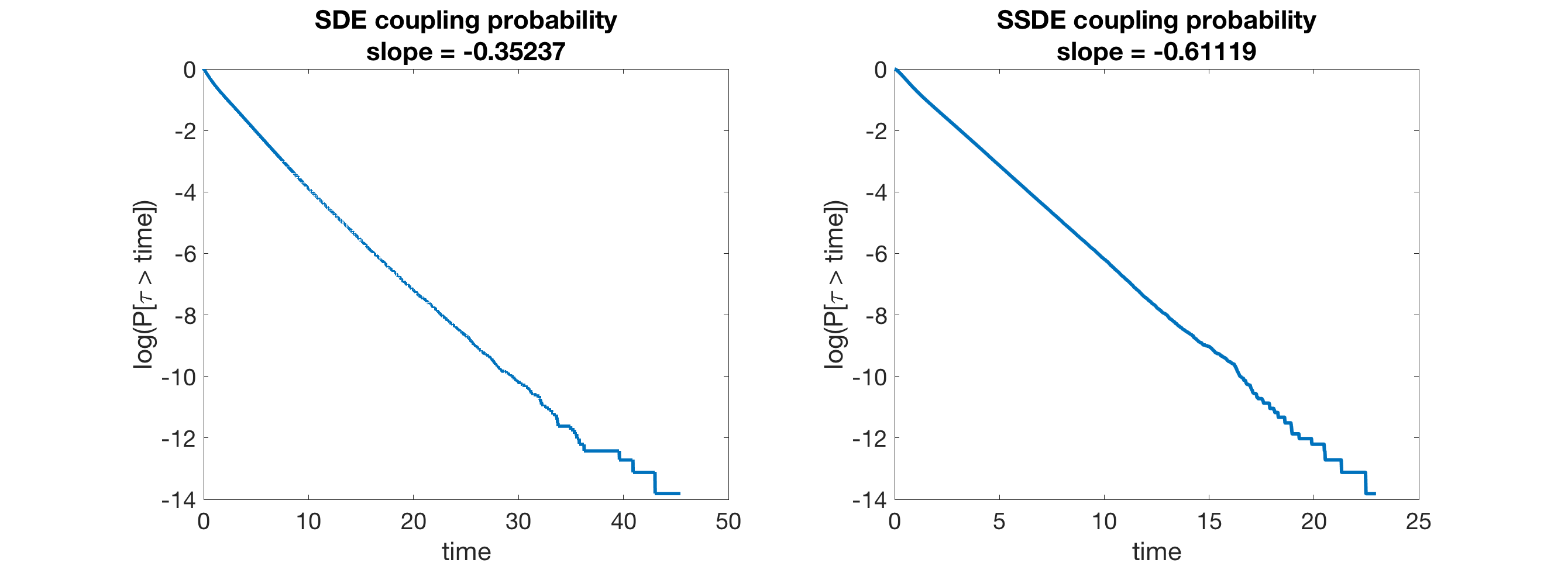}
 		\caption{Sensitivity analysis for the Lotka--Volterra
                  model. Note the exponential tails of the coupling time for the
                  SDE and SSDE models. }
\label{fig4}\end{center}
 \end{figure}

\section{Beddington-DeAngelis predator-prey dynamics}\label{s:BD}
\subsection{ODE}
In 1975 \cite{B75} and \cite{D75} came up with a new functional response to better explain some predator-prey interactions. Since then, these models have been used extensively in ecology.
The system is governed by
\begin{equation}\label{e:PP_ODE}
\begin{split}
   \frac{dX}{dt}(t) &= X(t)\left(a_1-b_1X(t) - \frac{c_1Y(t)}{m_1+m_2X(t)+m_3Y(t)}\right) \\
     \frac{dY}{dt}(t) &=  Y(t)\left(-a_2-b_2Y(t)+ \frac{c_2X(t)}{m_1+m_2X(t)+m_3Y(t)}\right)
\end{split}
\end{equation}
The dynamics has been classified in the deterministic setting in a series of papers \citep{CC01, H11, LT11}.
The following result by \cite{LT11} shows when one has persistence (or permanence).
\begin{thm}\label{t:pp_ode}
Suppose one of the following two conditions is satisfied
\begin{itemize}
  \item $c_2 > a_2 m_2$ and $(c_2-a_2m_2)\left(\frac{a_1}{b_1}-\frac{c_1}{b_1m_3}-\frac{a_2 m_3}{b_2m_2}\right)>a_2 m_1$
  \item $a_1m_3>c_1+\frac{b_1a_2m_3^2}{b_2m_2}$ and $(c_2-a_2m_2)\left(\frac{a_1}{b_1}-\frac{c_1}{b_1m_3}-\frac{a_2 m_3}{b_2m_2}\right)>a_2 m_1$.
  Then the system \eqref{e:PP_ODE} is permanent, i.e., there exist $0<\ell\leq L$ such that
  \[
  \min\left\{\liminf_{t\to\infty} X(t),\liminf_{t\to\infty} Y(t)\right\}\geq \ell, ~~\max\left\{\limsup_{t\to\infty} X(t), \limsup_{t\to\infty}Y(t)\right\}\leq L.
  \]
\end{itemize}
\end{thm}
Moreover, if there is no intraspecific competition for the predators, $b_2=0$, then by \cite{H03} we have the following result.
\begin{thm}\label{t:pp_ode2}
The system \eqref{e:PP_ODE} with $b_2=0$ has a unique limit cycle if and only if
\[
\frac{m_2a_2}{c_2}<\left(1+ \frac{m_1b_1}{m_2a_1}\right)^{-1}
\]
and
\[
\frac{\frac{c_1}{m_3a_1}-\frac{c_2}{m_1a_1}}{x_*+y_*+\frac{m_1b_1}{m_2a_1}}> \frac{x_*+y_*+\frac{m_1b_1}{m_2a_1}}{y_*}
\]
where $(x_*, y_*)$ are the positive solutions to
\[
1-x_* - \frac{\frac{c_1}{m_3a_1}y_*}{x_*+y_*+\frac{m_1b_1}{m_2a_1}} = 0, \frac{x_*}{x_*+y_*+\frac{m_1b_1}{m_2a_1}} = \frac{m_2a_2}{c_2}.
\]
\end{thm}

\subsection{SDE}
\cite{DDY16} analyze the SDE equivalent of \eqref{e:PP_ODE}. If $X(t), Y(t)$ denote the prey and predator densities, the dynamics follows the SDE
\begin{equation}\label{e:PP_SDE}
\begin{split}
   dX(t) &= X(t)\left(a_1-b_1X(t) - \frac{c_1Y(t)}{m_1+m_2X(t)+m_3Y(t)}\right)\,dt+\alpha X(t) dB_1(t) \\
     dY(t) &=  Y(t)\left(-a_2-b_2Y(t)+ \frac{c_2X(t)}{m_1+m_2X(t)+m_3Y(t)}\right)\,dt+\beta Y(t) dB_2(t)
\end{split}
\end{equation}
If $$\lambda_X(\delta_0)=a_1-\frac{\alpha^2}{2}>0$$ species $X$ persists on its own and has a unique stationary distribution $\mu_x$ on $(0,\infty)$. In order to study the coexistence of the two species we have to look at the Lyapunov exponent
\[
\lambda_Y(\mu_x) = -a_2-\frac{\beta^2}{2} + \int_{(0,\infty)} \frac{c_2 x}{m_1+m_2 x}\mu_x(dx) =  -a_2-\frac{\beta^2}{2} + C\int_{(0,\infty)} \frac{c_2x^{q}e^{-ax}}{m_1+m_2x}
\]
where $a:=\frac{2b_1}{\alpha^2}>0, q:=\frac{2a_1}{\alpha^2}-1>0$ and $C$ is a normalizing constant. The following result, which follows from Theorem \ref{t:p_sde} (or \cite{DDY16}), tells us when the species coexist.
\begin{thm}\label{t:pp_sde}
Suppose $\lambda_X(\delta_0)=a_1-\frac{\alpha^2}{2}>0$ and
\[
\lambda_Y(\mu_x)= -a_2-\frac{\beta^2}{2} + C\int_{(0,\infty)} \frac{c_2x^{q}e^{-ax}}{m_1+m_2x}>0.
\]
Then both species persist and the process converges exponentially fast to the unique invariant probability measure $\mu^*$ on $\R_+^{2,\circ}$.
\end{thm}

\subsection{SSDE}
We analyze the SSDE version of \eqref{e:PP_SDE}.
Let $X(t), Y(t)$ denote the prey and predator densities and assume $r(t)$ is an independent irreducible Markov chain with stationary distribution $\nu=(\nu_1,\dots,\nu_{n_0})$. We assume the dynamics is

\begin{equation}\label{e:PP_SDE}
\begin{split}
   dX(t) &= X(t)\left(a_1(r(t))-b_1(r(t))X(t) - \frac{c_1(r(t))Y(t)}{m_1(r(t))+m_2(r(t))X(t)+m_3(r(t))Y(t)}\right)\,dt\\
   &~~~+\alpha(r(t)) X(t) dB_1(t) \\
     dY(t) &=  Y(t)\left(-a_2(r(t))-b_2(r(t))Y(t)+ \frac{c_2(r(t))X(t)}{m_1(r(t))+m_2(r(t))X(t)+m_3(r(t))Y(t)}\right)\,dt\\
     &~~~+\beta(r(t)) Y(t) dB_2(t)
\end{split}
\end{equation}
In this setting species $X$ survives and on its own if
\[
\lambda_x(\delta_0\times \nu) = \sum_{k=1}^{n_0}\nu_i\left(a_1(k)-\frac{\alpha^2(k)}{2}\right)>0.
\]
The process $(X,r)$ converges on $(0,\infty)\times \CN$ to a probability measure $\mu^{x,r}$. The coexistence of the predator and the prey are then determined by the Lyapunov exponent
\[
\lambda_y(\mu^{x,r}) = \sum_{k=1}^{n_0}\nu_i\left(-a_2(k)-\frac{\beta^2(k)}{2}\right) +  \sum_{k=1}^{n_0}\int_0^\infty \frac{c_2(k)x}{m_1(k)+m_2(k)x}\mu^{x,r}(dx,k).
\]
\begin{thm}\label{t:pp_ssde}
Suppose $$\lambda_x(\delta_0\times \nu) = \sum_{k=1}^{n_0}\nu_i\left(a_1(k)-\frac{\alpha^2(k)}{2}\right)>0$$ and
\[
\lambda_y(\mu^{x,r}) = \sum_{k=1}^{n_0}\nu_i\left(-a_2(k)-\frac{\beta^2(k)}{2}\right) +  \sum_{k=1}^{n_0}\int_0^\infty \frac{c_2(k)x}{m_1(k)+m_2(k)x}\mu^{x,r}(dx,k)>0.
\]
Then both species persist and the process $(X(t), Y(t), r(t))$ converges exponentially fast to the unique invariant probability measure $\mu^*$ on $\R_+^{2,\circ}\times \CN$.
\end{thm}
\begin{proof}
This is an immediate application of Theorem \ref{t:p_ssde}.
\end{proof}
We note that this result simplifies the one from \cite{BS16} where the authors have a more complicated expression $\bar \lambda$ for the Lyapunov exponent and did not realize that $\bar \lambda = \lambda_y(\mu^{x,r})$.

\subsection{Numerical examples}

We analyze numerically the Beddington-DeAngelis predator-prey
model. Two sets of parameters are considered. The first set of
parameters is $a_{1} = 8, a_{2} = 1, b_{1} = 1.1, b_{2} = 0, c_{1} =
10, c_{2} = 4, m_{1} = 2, m_{2} = 2, m_{3} = 0.5$. The second
set of parameters is $a_{1} = 3.9, a_{2} = 1.2, b_{1 } = 0.5, b_{2} = 0, c_{1} = 4, c_{2}
= 4, m_{1} = 2, m_{2} = 1, m_{3} = 0.3$. It follows from Theorem
\ref{t:pp_ode2} that the deterministic model for each set of parameters admits a unique limit cycle. The
strength of the white noise fluctuations is chosen to be $\alpha = \beta =
0.35$.

1) Figure \ref{fig1BD} Top shows the limit cycles of the
Beddington-DeAngelis ODE model for the two sets of parameter we work with. The invariant probability measures for the SDE models are found on a numerical domain $(0,
8)^{2}$ and shown in Figure \ref{fig1BD} Bottom. The mesh size of our numerical computation is $1200 \times
1200$. The reference solution is obtained from a Monte Carlo simulation
with $10^{8}$ sample points. We can see that although the predator does not go extinct, the invariant probability measure
concentrates significantly near the $Y$-axis.

\textit{\textbf{Biological Interpretation:} The white noise fluctuations perturb the dynamics of the two species. The invariant probability measures still look qualitatively like concentric closed trajectories. The environmental fluctuations create a very interesting phenomenon. The densities of the invariant probability measure are highest close to the $y$ axis. This shows that in this system, the prey can become low and stay low for a long time, while the predator has a significant density and takes a long time to die out in the absence of a food source. When the predator population finally decreases, the prey increases again, which causes after some time the predator to increase. This cycle would then get repeated. Our results highlight that even though the theory implies coexistence, a real system might go extinct: spending a long time close to the boundary will make it more likely to have extinction induced by demographic stochasticity.
}

2) We next consider the PDMP version of the model when there are two
environmental states. Parameter sets are the same as in the SDE case. We
compute the invariant probability measure for three different
switching rates, from slow to fast. The switching rates are
$q_{12} =  q_{21} = 2.5$ for the slow switching rate, $q_{12}
= q_{21} = 5$ for the medium switching rate, and $q_{12} = q_{21} = 10$
for the fast switching rate. We have seen that for each set of parameters we get an ODE that has a limit
cycle. Furthermore, the two limit cycles are different. When the rate of
switching increases, the two marginal invariant probability measures from the two possible environmental states move closer to each other. This new
invariant probability measure does not seem to
concentrate on either of the two limit cycles. The result is
shown in Figure \ref{fig2BD}. The numerical domain is still
$(0, 8)^{2}$ with a mesh size $800 \times 800$. The reference solution is obtained from a Monte Carlo
simulation with $10^{8}$ sample points.

\textit{\textbf{Biological Interpretation:} Switching between the two environments creates a coexistence situation where the invariant measure is qualitatively similar to the occupation measure of a limit cycle. However, in this setting the support of the stationary distribution seems to be bounded away from the extinction set - there is no concentration near the extinction set. As the switching speed increases the invariant probability measures from the two environments become closer and closer. The dynamics with fast switching is close to an ODE system and will have a unique limit cycle that is significantly larger than the limit cycles from the two separate environments. Environmental fluctuations can significantly change the dynamics and make the species densities oscillate at greater amplitudes (more than double the amplitudes from each fixed environment).
}

3) The third numerical simulation involves the SSDE
model. The model parameters (including the switching rates) are the same as in the PDMP model. In addition, we take the environmental fluctuation strengths
to be $\alpha = \beta = 0.35$. Again, we run the simulation
for three different switching rates. The switching rates are
$q_{12} =  q_{21} = 1$ for the slow switching rate, $q_{12}
= q_{21} = 2.5$ for the medium switching rate, and $q_{12} = q_{21} = 10$
for the fast switching rate. We get the existence of a unique
invariant probability measure which is absolutely continuous with
respect to Lebesgue measure - see Figure \ref{fig3BD}.  Again, the
probability density function significantly concentrates near the
$Y$-axis. We cap the heat map at $0.05$ to make the probability
density function visible in the area with lower probability density. The numerical domain is
$(0, 8)^{2}$ with a mesh size $800 \times 800$. The reference solution is obtained from a Monte Carlo
simulation with $10^{8}$ sample points.

\textit{\textbf{Biological Interpretation:} The combination of random switching and white noise make the coexistence measure spread out and concentrate close to the $Y$ axis. As in the SDE example, this shows that in this setting one might not be entitled to neglect demographic stochasticity.
}

The last step is the sensitivity analysis. We use Milstein scheme to
simulate SDE and SSDE. The estimator of the mean finite time
error is still
$$
  I = \frac{1}{N}\sum_{i = 0}^{N - 1}\rho(\hat{X}^{dt}_{i,T},
  \hat{X}^{2dt}_{i,T}) \,.
$$
Parameters in our simulation are $dt = 0.001$ and $N = 10^{6}$. For the SDE
model, we only run simulations using the first parameter set. The time
span is $T = 5$ for SDE, and $T = 6$ for SSDE. The estimator $I$ is
$0.0398$ for SDE, and $0.0404$ for SSDE. We next run coupling method to show the speed
of convergence. The exponential tails of coupling times are demonstrated in
Figure \ref{fig4BD}. We can see that the coupling time distribution has
 exponential tails in each case. This gives us an estimate of
$\alpha \approx 0.395$ for the SDE and $\alpha \approx 0.467$ for SSDE. Therefore, we
have $\mathrm{d}_{w}( \mu^{*}, \hat{\mu}) \approx 0.0658$ for the SDE, and
$\mathrm{d}_{w}( \mu^{*}, \hat{\mu}) \approx 0.0758$ for the
SSDE. These error terms are larger than in the case of the Lotka-Volterra
model, because the coupling time is slower in the presence of a limit
cycle. Two trajectories in the coupled SDE need to chase each other
near the limit cycle in order to get close. Nevertheless, the numerical
results still show that the obtained invariant probability
measures are acceptable in both the Lotka-Volterra and the Beddington-DeAngelis cases.

 \begin{figure}[!htbp]
 	\begin{center}
 		\includegraphics[width = \linewidth]{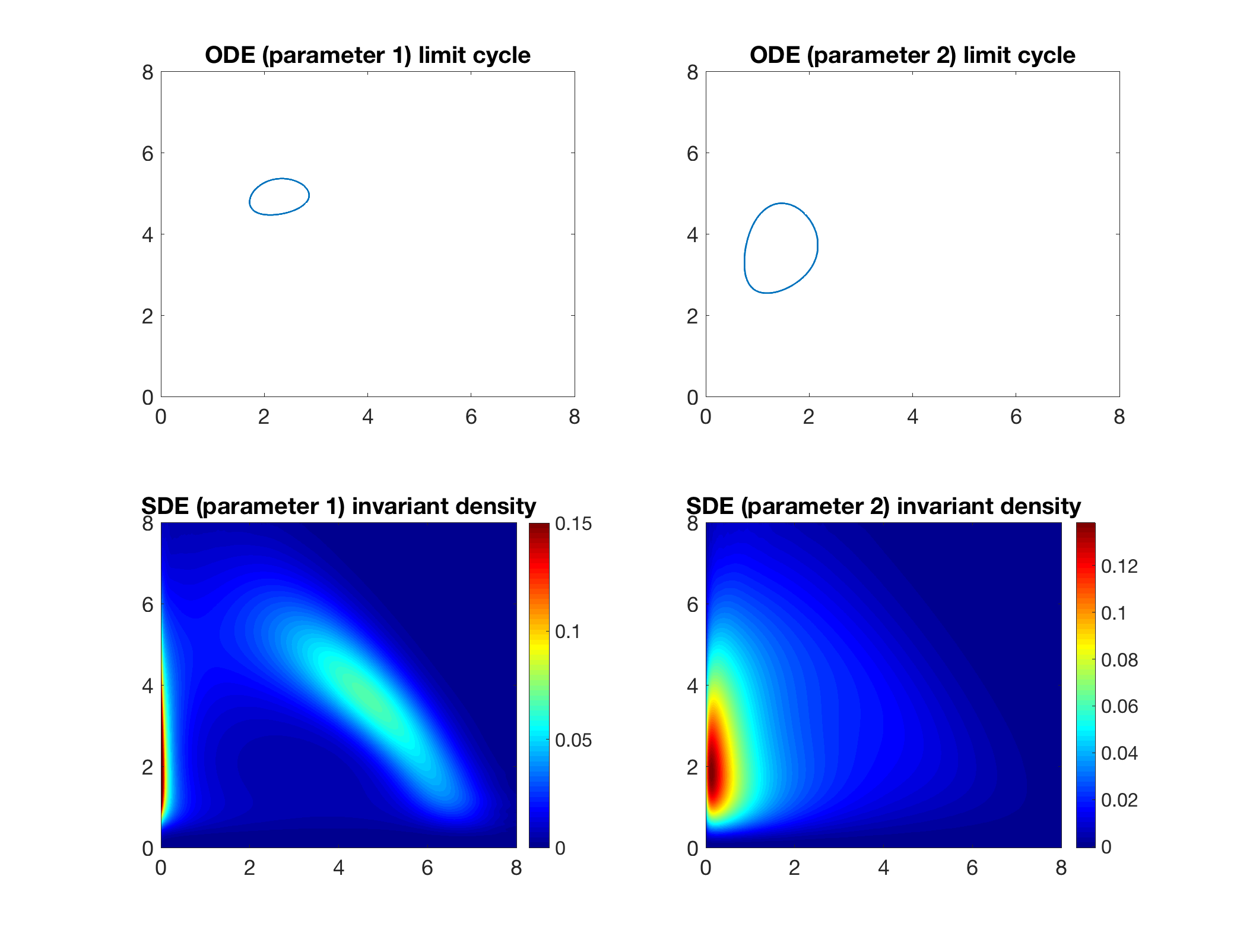}
 		\caption{{\bf Top}: Trajectories of $X_{1}(t)$ and
                  $X_{2}(t)$ for the two parameter sets. {\bf Bottom}:
                  Invariant probability measures for the two parameter
                  sets. }
\label{fig1BD}\end{center}
 \end{figure}

\begin{figure}[!htbp]
 	\begin{center}
 		\includegraphics[width = \linewidth]{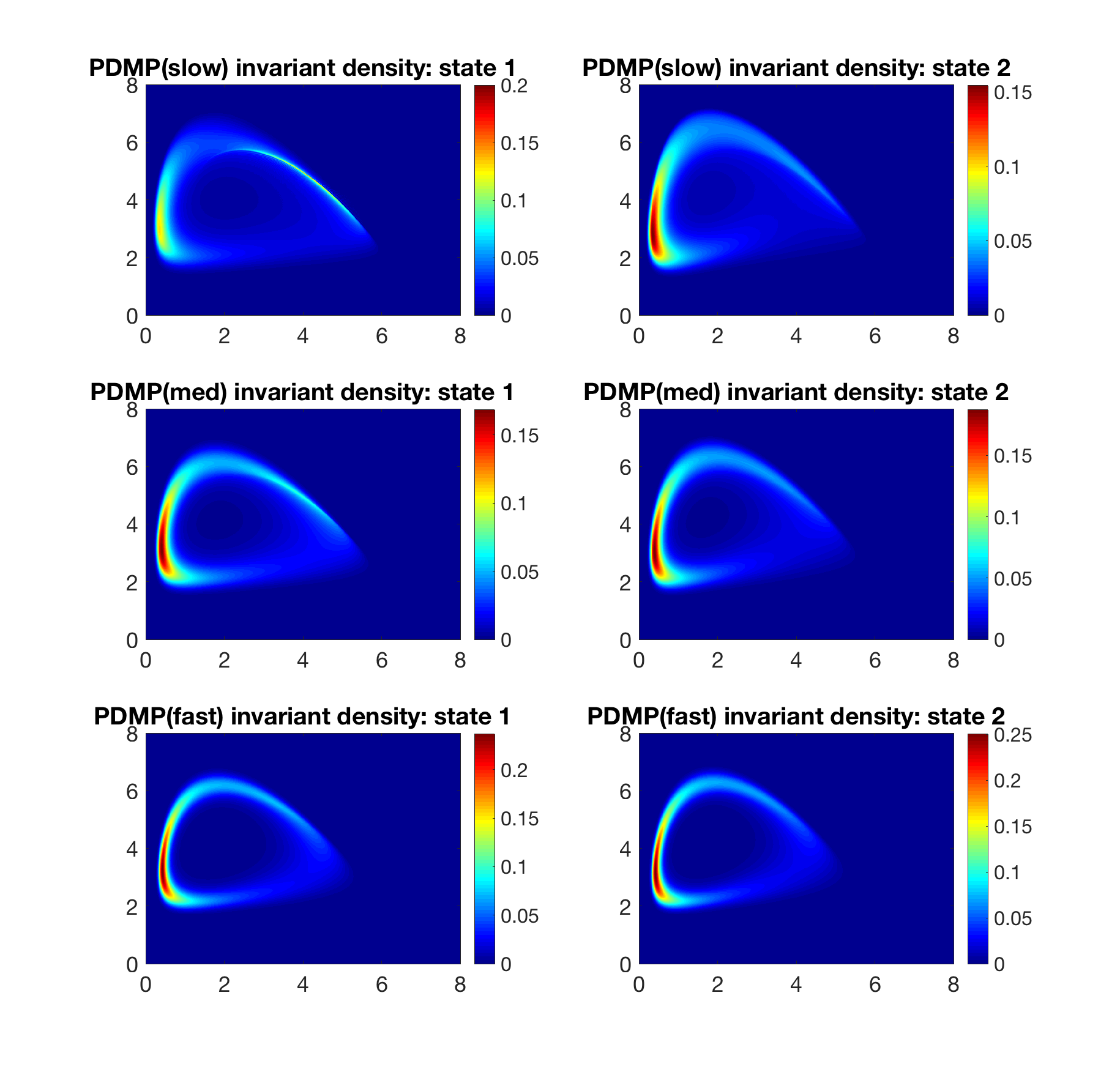}
 		\caption{Top to bottom: Invariant probability measures
                of the PDMP model for three different sets of switching rates. }
\label{fig2BD}\end{center}
 \end{figure}

\begin{figure}[!htbp]
 	\begin{center}
 		\includegraphics[width = \linewidth]{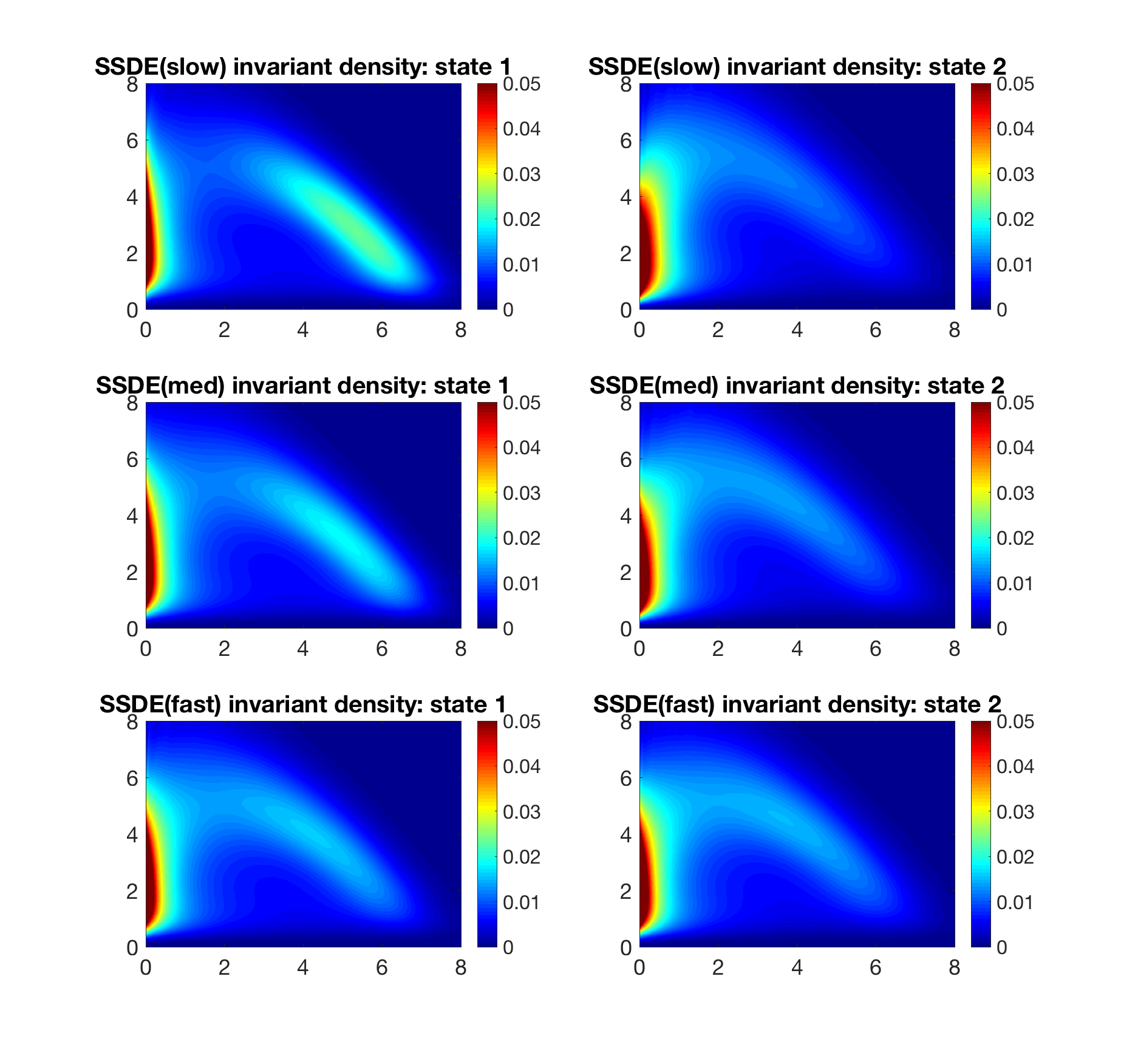}
 		\caption{Top to bottom: Invariant probability measures
                of the SSDE model for three different sets of switching rates. }
\label{fig3BD}\end{center}
 \end{figure}

\begin{figure}[!htbp]
 	\begin{center}
 		\includegraphics[width = \linewidth]{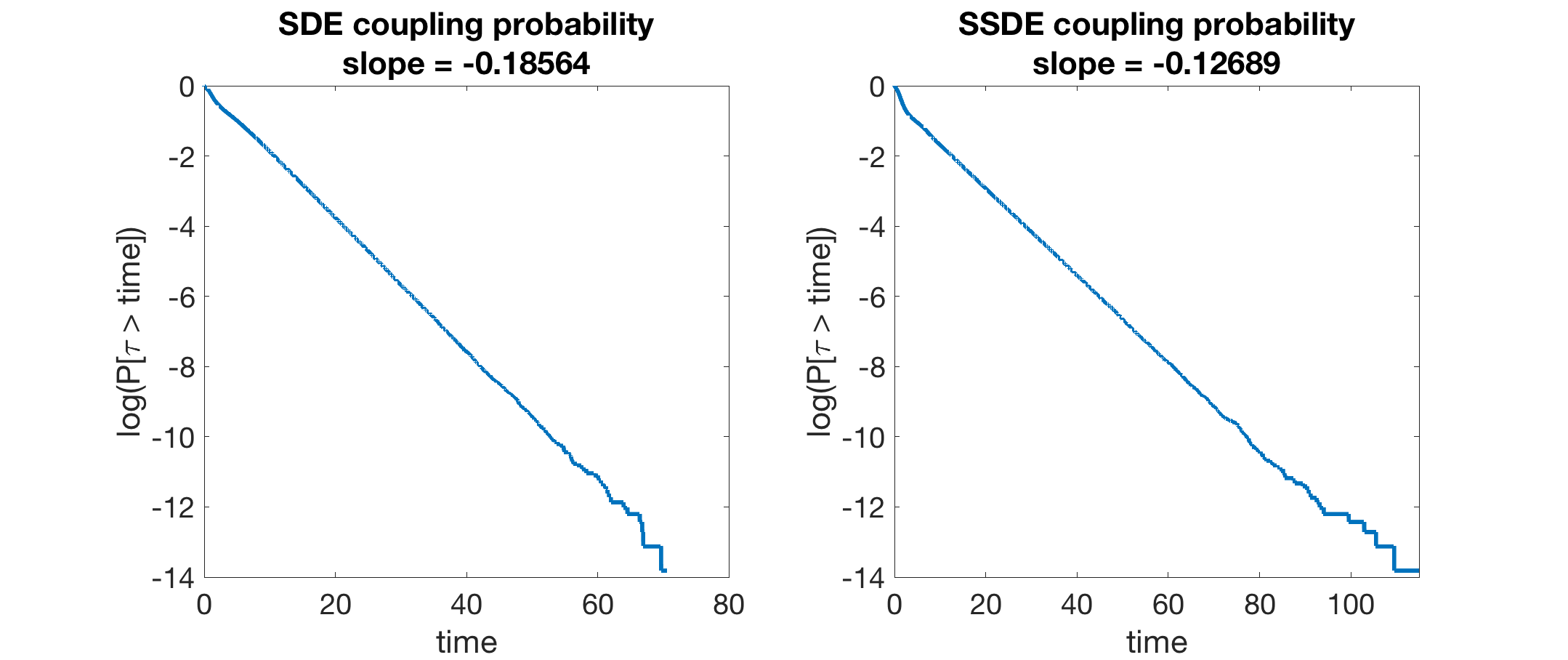}
 		\caption{Exponential tails of the coupling time for the
                  SDE and SSDE models. }
\label{fig4BD}\end{center}
 \end{figure}

\section{Rock-paper-scissors dynamics}\label{s:RPS}
It has been observed in nature that certain three species systems can have rock-paper-scissors dynamics. One such example is the side-blotched lizard \cite{SL96}. There are three different types of lizards. The first type is a highly aggressive lizard that attempts to control a large area and mate with any females within the area. The second type is a furtive lizard, which wins against the aggressive lizard by acting like a female. This way the furtive lizard cand mate without being detected in an aggressive lizard’s territory. The third type is a guarding lizard that watches one specific female for mating. This prevents the furtive lizard from mating. However, the guarding lizard is not strong enough to overcome the aggressive lizard.

Interestingly this type of dynamics creates regimes where one species seems to win, until the species that beats it makes a comeback. This creates subtle technical problems.

\subsection{ODE}
Pick $0<\beta<1<\alpha$ and set
  \begin{equation}\label{e:ODE}
\begin{split}
d\bar X_1(t) &= \bar X_1(t)\left(1-\bar X_1(t)-\alpha \bar X_2(t)-\beta \bar X_3(t)\right)\,dt \\
d\bar X_2(t) &= \bar X_2(t)\left(1-\beta X_1(t)-\bar X_2(t)-\alpha\bar X_3(t)\right)\,dt \\
d\bar X_3(t) &= \bar X_3(t)\left(1-\alpha \bar X_1(t) -\beta \bar X_2(t) - \bar X_3(t)\right)\,dt \\
\end{split}
\end{equation}
This is the model introduced by \cite{ML75}. Let $\Delta =\{\bx\in \R_+^3~:~x_1+x_2+x_3=1\}$ be the unit simplex. One can see that \eqref{e:ODE} has five fixed points. The origin $0$ is a source, the canonical basis vectors $e_1, e_2, e_3$ are saddle points and the equilibrium that is not on the boundary is given by
\[
\bar x = \frac{e_1+e_2+e_3}{1+\alpha+\beta}
\]
Let $W^S(A)$ be the stable manifold of the set $A$ and let $D=\{\bx\in\R_+^3~:~x_1=x_2=x_3\}$ be the diagonal. One can show that
\[
\Omega : = W^S(e_1) \cup W^S(e_2) \cup W^S(e_3)
\]
is a heteroclinic cycle. We have the following description of the dynamics (\cite{HS98, B18}):
\begin{enumerate}
  \item If $\alpha+\beta<2$ the interior fixed point $\bar x$ is a sink and all trajectories starting in $\R_+^{3,\circ}$ converge to $\bar x$.
  \item If $\alpha+\beta>2$ the interior point $\bar x$ is a saddle with stable manifold $D\setminus\{0\}$. Every trajectory starting from $\R_+^{3,\circ}\setminus D$ has $\Omega$ as omega limit cycle.
  \item If $\alpha+\beta=2$ the set $\Delta$ is invariant and attracts all nonzero trajectories, $\Omega=\partial \Delta$ and trajectories starting in $\Delta^\circ\setminus \{\bar x\}$ are periodic.
\end{enumerate}

\subsection{SDE}
Assume the system is given by the stochastic differential equations
\begin{equation}\label{e:LV}
dX_i(t)=X_i(t)\left(\mu_i-\sum_{j=1}^3 a_{ij} X_j(t)\right)\,dt +  \sigma_i X_i(t)\,dB_i(t), X_i(0)=x_i\geq 0
\end{equation}
where $B_1, B_2, B_3$ are independent Brownian motions.
The constant $\mu_i$ is the per-capita growth rate of the the $i$th species, and $a_{ij}>0$ is the coefficient measuring the competition strength of species $j$ on species $i$. Set $$\bar \mu_i:= \mu_i-\frac{\sigma_{i}^2}{2}.$$
Each species persists on its own, so $\bar \mu_i>0$ and there exists a unique invariant probability measure $\mu_i^*$ on $(0,\infty)$.
The Lyapunov exponents can be computed as
\begin{equation}\label{e:lambda}
\lambda_i(\mu^*_j) = \frac{\bar \mu_i a_{jj}-a_{ij}\bar \mu_j}{a_{jj}}.
\end{equation}
We assume that species $2$ outcompetes species $1$, species $3$ outcompetes species $2$ and species $1$ outcompetes species $3$:
$$\lambda_2(\mu_1^*)>0, \lambda_3(\mu_1^*)<0, \lambda_1(\mu_2^*)<0, \lambda_3(\mu_2^*)>0, \lambda_1(\mu_3^*)>0, \lambda_2(\mu_3^*)<0.$$

In order to be in the rock-paper-scissors setting we therefore need

\[
\frac{a_{12}}{a_{22}} > \frac{\bar \mu_1}{\bar \mu_2}, \frac{a_{11}}{a_{21}}> \frac{\bar \mu_1}{\bar \mu_2},
\]

\[
\frac{a_{23}}{a_{33}} > \frac{\bar \mu_2}{\bar \mu_3}, \frac{a_{22}}{a_{32}}> \frac{\bar \mu_2}{\bar \mu_3},
\]

and

\[
\frac{a_{31}}{a_{11}} > \frac{\bar \mu_3}{\bar \mu_1}, \frac{a_{33}}{a_{13}}> \frac{\bar \mu_3}{\bar \mu_1},
\]
The following results quantify when one has persistence or extinction for this dynamics (see \cite{HNS20} for complete proofs).
\begin{thm} [Persistence]\label{t:pers}
If the product of the Lyapunov exponents (invasion rates) pushing the process away from the boundary $|\lambda_2(\mu_1^*)\lambda_3(\mu_2^*)\lambda_1(\mu_3^*)|$ is strictly greater than the product of the Lyapunov exponents attracting the process towards the boundary $|\lambda_3(\mu_1^*)\lambda_1(\mu_2^*)\lambda_2(\mu_3^*)|$ we get persistence and exponential convergence to an invariant probability measure $\mu^*$ on $\R_+^{3,\circ}$.
\end{thm}
\begin{thm}
[Extinction]\label{t:ext} If the product of the Lyapunov exponents (invasion rates) pushing the process away from the boundary $|\lambda_2(\mu_1^*)\lambda_3(\mu_2^*)\lambda_1(\mu_3^*)|$ is strictly smaller than the product of the Lyapunov exponents attracting the process towards the boundary $|\lambda_3(\mu_1^*)\lambda_1(\mu_2^*)\lambda_2(\mu_3^*)|$ we get extinction, in the sense that there exists $\alpha>0$ such that

\[
\PP_\bx\left(\limsup_{t\to\infty} \frac{\dist(\BX(t),\partial\R^3_+)}{t}<\alpha\right)=1
\]
for any $\bx\in\R^{3,\circ}_+$
\end{thm}

Using \eqref{e:lambda} we get the following corollaries.
\begin{cor}\label{c:1}
If \[
(\bar \mu_2 a_{11}-a_{21}\bar \mu_1)(\bar \mu_3 a_{22}-a_{32}\bar \mu_2)(\bar \mu_1 a_{33}-a_{13}\bar \mu_3) > |(\bar \mu_1 a_{22}-a_{12}\bar \mu_2)(\bar \mu_2 a_{33}-a_{23}\bar \mu_3)(\bar \mu_3 a_{11}-a_{31}\bar \mu_1)|
\]
then the system persists and converges exponentially fast to the invariant probability measure $\mu^*$ on $\R_+^{3,\circ}$.
\end{cor}

\begin{cor}\label{c:2}
If \[
(\bar \mu_2 a_{11}-a_{21}\bar \mu_1)(\bar \mu_3 a_{22}-a_{32}\bar \mu_2)(\bar \mu_1 a_{33}-a_{13}\bar \mu_3) < |(\bar \mu_1 a_{22}-a_{12}\bar \mu_2)(\bar \mu_2 a_{33}-a_{23}\bar \mu_3)(\bar \mu_3 a_{11}-a_{31}\bar \mu_1)|
\]
then the system goes extinct, i.e. with probability one
\[
\BX(t)\to \partial \R_+^{3,\circ}
\]
\end{cor}
By Theorem \ref{t:ext} we have extinction and convergence to the boundary when
\[
(\bar \mu_2 a_{11}-a_{21}\bar \mu_1)(\bar \mu_3 a_{22}-a_{32}\bar \mu_2)(\bar \mu_1 a_{33}-a_{13}\bar \mu_3) < |(\bar \mu_1 a_{22}-a_{12}\bar \mu_2)(\bar \mu_2 a_{33}-a_{23}\bar \mu_3)(\bar \mu_3 a_{11}-a_{31}\bar \mu_1)|
\]
\begin{thm}\label{t:rps}
  Let $0<\beta<1<\alpha$ and define the process $(X_1(t), X_2(t), X_3(t))$ by
\begin{equation}\label{e:SDE2}
\begin{split}
dX_1(t) &= X_1(t)\left(\mu_1-X_1(t)-\alpha X_2(t)-\beta X_3(t)\right)\,dt +  \sigma_1 X_1(t) \,dE_1(t)\\
dX_2(t) &= X_2(t)\left(\mu_2-\beta X_1(t)-X_2(t)-\alpha X_3(t)\right)\,dt + \sigma_2 X_2(t)\,dE_2(t)\\
dX_3(t) &= X_3(t)\left(\mu_3-\alpha X_1(t) -\beta X_2(t) - X_3(t)\right)\,dt + \sigma_3 X_3(t) \,dE_3(t).\\
\end{split}
\end{equation}
Suppose furthermore that $\bar \mu_1 = \bar \mu_2 = \bar \mu_3>0$. If $2>\alpha+\beta$ the system persists and converges exponentially fast to a unique invariant probability measure $\mu^*$ on $\R_+^{3,\circ}$. If $2<\alpha+\beta$ then the system goes extinct in the sense that $\BX(t)\to \partial \R_+^3$.
\end{thm}
\begin{proof}
If $2>\alpha+\beta$, we have
\begin{equation}\label{e:SDE2}
\begin{split}
(\bar \mu_2 a_{11}-a_{21}\bar \mu_1)(\bar \mu_3 a_{22}-a_{32}\bar \mu_2)(\bar \mu_1 a_{33}-a_{13}\bar \mu_3) &= \bar \mu_1^3(1-\beta)^3\\
&>\bar \mu_1^3(\alpha-1)^3\\
&= |(\bar \mu_1 a_{22}-a_{12}\bar \mu_2)(\bar \mu_2 a_{33}-a_{23}\bar \mu_3)(\bar \mu_3 a_{11}-a_{31}\bar \mu_1)|
\end{split}
\end{equation}
and by Corollary \ref{c:1} there is persistence and exponential convergence to a unique invariant probability measure. If one has instead that $2<\alpha+\beta$ then  Corollary \ref{c:2} implies that with probability one $\BX(t)\to \partial \R_+^3$.
\end{proof}

\subsection{PDMP}
Suppose the environment is modelled by $r(t)$ and switches between the two states $1$ and $2$ with the rate matrix
\[
  \left( {\begin{array}{cc}
   -\tau(1-p) & \tau(1-p) \\
   \tau p & -\tau p \\
  \end{array} } \right)
\]
for some $p\in (0,1)$ and $\tau>0$. This means that if we start in environment 1, we wait for an exponential time with parameter $\tau(1-p)$, switch to environment 2, wait for an independent exponential time with parameter $\tau p$, switch to environment 1 and repeat this process indefinitely.
Let $(\alpha_1, \beta_1)$ and $(\alpha_2,\beta_2)$ be two parameters such that
\[
\alpha_1+\beta_1>2,
\]
and
\[
\alpha_2+\beta_2<2.
\]
In environment $1$ there is extinction while in environment $2$ there is persistence. We are interested in the dynamics of the switching process
  \begin{equation}\label{e:PDMP}
\begin{split}
\frac{dX_1}{dt}(t)&=X_1(t)(1-X_1(t)-\alpha_{r(t)} X_2(t) - \beta_{r(t)} X_3(t))\\
 \frac{dX_2}{dt}(t)&=X_2(t)(1-\beta_{r(t)} X_1(t) - X_2(t) - \alpha_{r(t)} X_3(t))\\
 \frac{dX_3}{dt}(t)&= X_3(t)( 1-\alpha_{r(t)} X_1(t)-\beta_{r(t)} X_2(t)-X_3(t))
\end{split}
\end{equation}
By \cite{B18} we have the following theorem.
\begin{thm}\label{t:rps_pdmp}
Let $(\alpha_1, \beta_1)$ and $(\alpha_2,\beta_2)$ be two parameters such that $\alpha_1+\beta_1>2$ and $\alpha_2+\beta_2<2$.
\begin{enumerate}
\item If
\[
\Lambda_b:=p(2-\alpha_1+\beta_1)+(1-p)(2-\alpha_2+\beta_2)>0
\]
then for $\tau$ small enough there is a unique persistent measure.
  \item If $\Lambda_b<0$ we have that $X(t)\to\partial \R_+^3$ with probability 1.
\end{enumerate}
\end{thm}

\subsection{SSDE} Let us now look at the SSDE setting. As before, let $(\alpha_1, \beta_1)$ and $(\alpha_2,\beta_2)$ be two parameters such that
\[
\alpha_1+\beta_1>2,
\]
and
\[
\alpha_2+\beta_2<2.
\]

Suppose $r(t)$ is a Markov chain that switches between states $1$ and $2$ and has a stationary distribution $(\nu_1,\nu_2)$. The dynamics will be
  \begin{equation}\label{e:SSDE_RPS}
\begin{split}
dX_1(t)&=X_1(t)(1-X_1(t)-\alpha_{r(t)} X_2(t) - \beta_{r(t)} X_3(t))\,dt +\sigma_1(r(t)) X_1(t)dB_1(t)\\
dX_2(t)&=X_2(t)(1-\beta_{r(t)} X_1(t) - X_2(t) - \alpha_{r(t)} X_3(t))\,dt+\sigma_2(r(t))X_2(t) dB_2(t)\\
dX_3(t)&= X_3(t)( 1-\alpha_{r(t)} X_1(t)-\beta_{r(t)} X_2(t)-X_3(t))\,dt +\sigma_3(r(t)) X_3(t)dB_3(t).
\end{split}
\end{equation}

\subsection{Numerical example}

In this section we look at the numerical solution of the
invariant probability density function for the rock-paper-scissors
model in the SDE and SSDE settings. We use the dynamics given by \eqref{e:LV} with
parameters $\mu  = 1.5, \alpha = 1.1, \beta = 0.4$ for the SDE model. The
magnitude of the white noise is taken to be $\sigma_{1} = \sigma_{2} =
0.5$.

In the SSDE model \eqref{e:SSDE_RPS} we take $\mu(2)  = 1.5, \alpha(2) = 1.1, \beta(2) = 0.4$ and $\mu(1) = 1.5,
\alpha(1) = 1.3, \beta(1) = 0.8$. The rate of random switching is $q_{12} = q_{21} = 1$. By Theorem \ref{t:rps} the three species can only coexist in environment
$2$. We see through simulations that in the SSDE setting all three species persist and
converge to a unique invariant probability measure.

In Figure \ref{fig9}, we analyze the SDE model on the domain $(0, 3)^{3}$ with $600
\times 600\times 600$ grids. The heat maps of six slices of the invariant probability
density function, at $z = 0.25, 0.5, 0.75, 1.0, 1.5,$ and $2$ respectively, are presented in Figure \ref{fig9}.

\textit{\textbf{Biological Interpretation:} As expected from the deterministic result, we can see in Figure \ref{fig9} that at the stationary distribution the species dynamics looks like a noisy heteroclinic cycle. The stationary distribution is supported on a compact set and has high densities close to the boundary. As we increase the density of the third species $z$, the support of the stationary distribution decreases - as $z$ goes from $0.25$ to $2$ the support of the stationary distribution decreases tenfold.
}

In Figure \ref{fig10}, we analyze the SSDE model on the domain $(0, 3)^{3} \times \{ 1, 2\}$ with $500 \times 500
\times 500 \times 2$ grids. The solutions of three slices of the
invariant probability function at both states, at $z = 0.12, 0.48$,
and $0.72$ respectively, are presented in Figure
\ref{fig10}. We use a 3D plot to show how the solution heavily concentrates close to the extinction set.

\textit{\textbf{Biological Interpretation:} If the dynamics switches between persistent and extinction prone rock-paper-scissors systems one can still get overall persistence. Even though the switching is random we expect the same phenomenon to arise when the switching is periodic. This would describe seasonality and annual cycles. In various regions which exhibit strong seasonal variation, organisms have developed survival mechanisms like hibernation and migration in order to survive seasons with resource shortages. Our results show that as long as the bad environment is not too harsh, the species can still persist.
}

We remark that it is very difficult to use the traditional method to solve a linear
system with this number of variables. However, our block
data-driven solver can significantly reduce the computational complexity. We first run a Monte Carlo
simulation with $8 \times 10^{9}$ samples. Then the big numerical
domain is divided into
$8000$ subdomains with $30 \times 30 \times 30$ (resp. $25\times
25\times 25 \times 2$) grids each for the SDE
model (resp. the SSDE model). We
then solve $8000$ optimization problems in parallel and then use the
``half-block shift'' technique proposed in \cite{dobson2019efficient} to reduce the
interface error between blocks. The entire computation takes about one hour on a
laptop.

Similar as the other two models, we carry out a sensitivity analysis
for the Rock-Paper-Scissor model. The Milstein scheme is used in our
simulations of both SDE and SSDE. As before, the estimator of the mean finite time
error equals
$$
  I = \frac{1}{N}\sum_{i = 0}^{N - 1}\rho(\hat{X}^{dt}_{i,T},
  \hat{X}^{2dt}_{i,T}) \,.
$$
Parameters in our simulation are $dt = 0.001$ and $N = 10^{6}$. For the SDE
model, we only run simulations using the first parameter set. The time
span is $T = 10$ for SDE, and $T = 50$ for SSDE. The estimator $I$ is
$0.00701$ for SDE, and $0.00571$ for SSDE. Then we run coupling method to show the speed
of convergence. The exponential tails of coupling times are showcased in
Figure \ref{fig11}. We can see that the coupling time distribution has
an exponential tail in each case. This gives us an estimate of
$\alpha \approx 0.5543$ for SDE and $\alpha \approx 0.4916$ for SSDE. Therefore, we
have $\mathrm{d}_{w}( \mu^{*}, \hat{\mu}) \approx 0.0157$ for the SDE, and
$\mathrm{d}_{w}( \mu^{*}, \hat{\mu}) \approx 0.0112$ for the
SSDE.

Our results show that the speed of convergence of the
rock-paper-scissors model is very slow, especially in the SSDE
setting. This is because a trajectory stays most of the time near
the boundary of the domain, where the magnitude of noise is
small. However, because the vector field does not change as rapidly as
in the other two models, the finite time error is well-controlled for
$T$ up to $50$. As a consequence, the numerical
results still imply that the obtained invariant probability
measures are accurate in both the SDE and SSDE settings.

\begin{figure}[!htbp]
\centerline{\includegraphics[width = \linewidth]{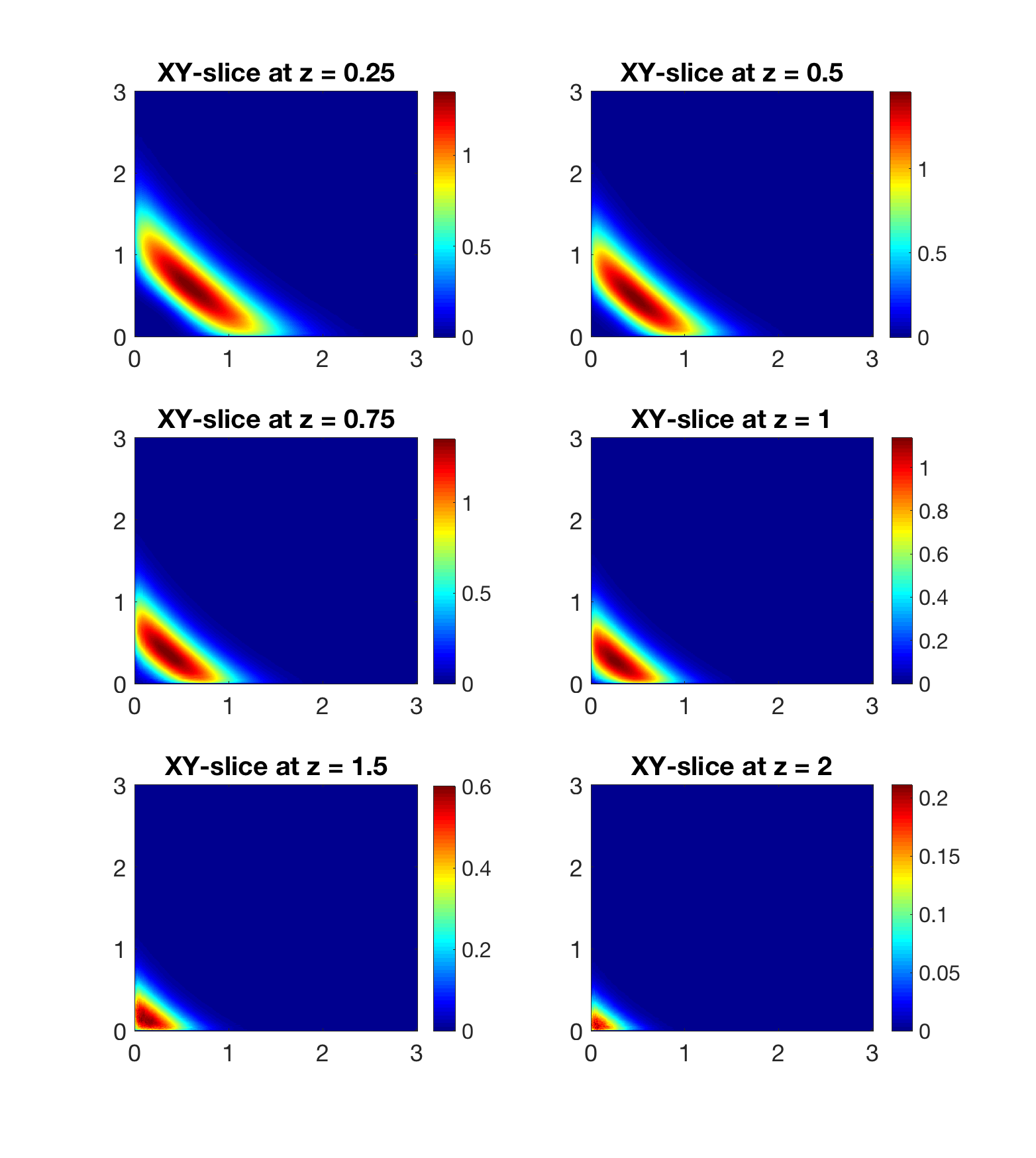}}
\caption{Invariant probability density function of the SDE rock-paper-scissors model. Parameters are $\mu  = 1.5, \alpha
= 1.1, \beta = 0.4, \sigma_{1} = \sigma_{2} =
0.5$. Top left to bottom right: Invariant
probability density function restricted on planes $z = 0.25, 0.5,
0.75, 1.0, 1.5,$ and $2$, respectively.}
\label{fig9}
\end{figure}

\begin{figure}[!htbp]
\centerline{\includegraphics[width = \linewidth]{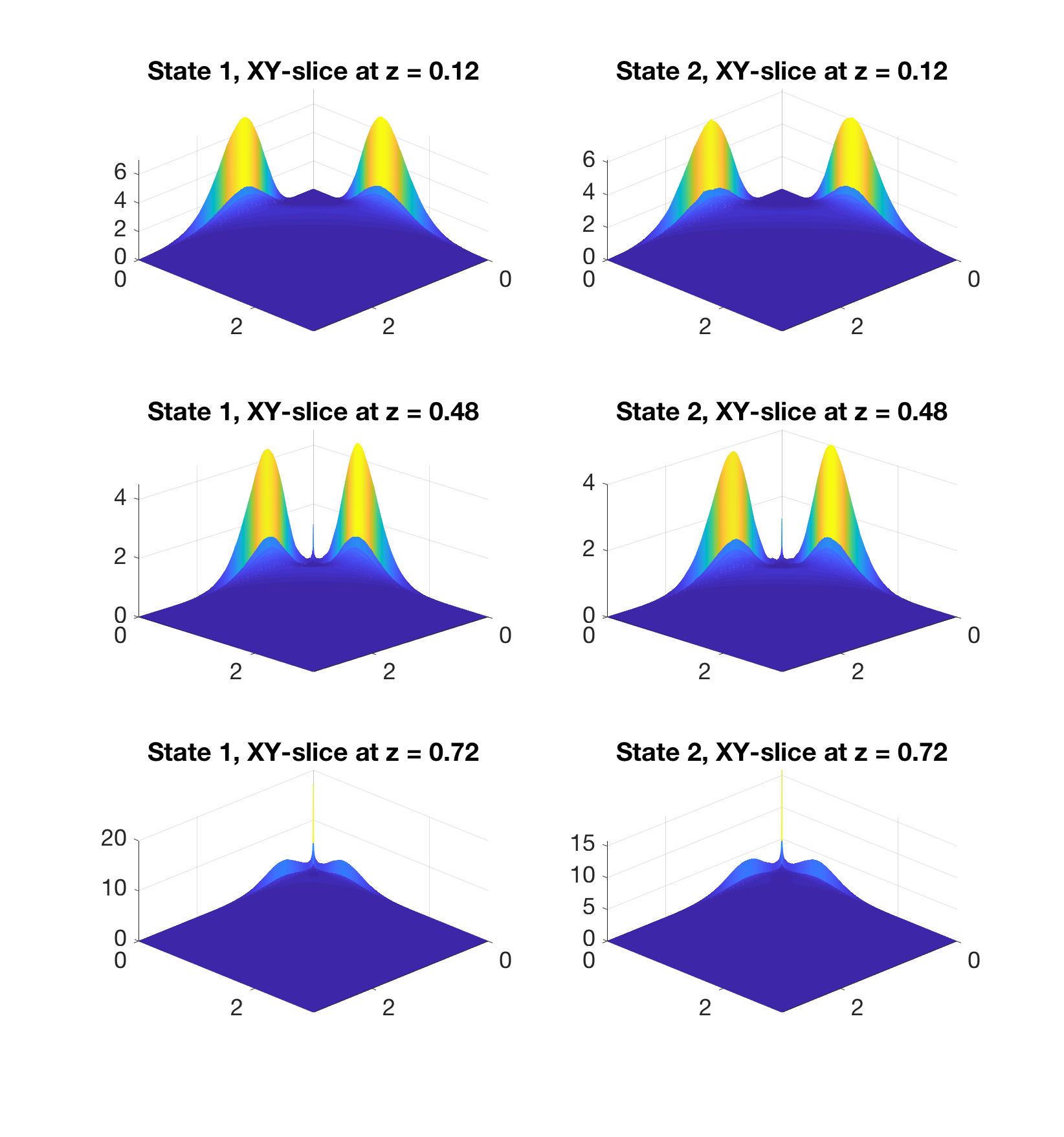}}
\caption{Invariant probability density function for the SSDE rock-paper-scissors model. Parameters are $\mu_{1} = \mu_{2} = \mu  = 1.5, \alpha(1) = 1.3, \alpha(2)
= 1.1, \beta(1) = 0.8, \beta(2) = 0.4, \sigma_{1} = \sigma_{2} =
0.5$, and $q_{12} = q_{21} = 1$. {\bf Top Left}: Trajectories of $(X_{1}(t),
X_{2}(t), X_{3}(t))$ up to $T  = 100$. {\bf Top}: Invariant
probability density function in both environmental states, restricted to the plane $\{ z =
0.12\}$. {\bf Middle}: Invariant
probability density function in both environmental states, restricted to the plane $\{ z =
0.48\}$.  {\bf Bottom}: Invariant
probability density function in both environmental states, restricted to the plane $\{ z =
0.72\}$.}
\label{fig10}
\end{figure}
\begin{figure}[!htbp]
 	\begin{center}
 		\includegraphics[width = \linewidth]{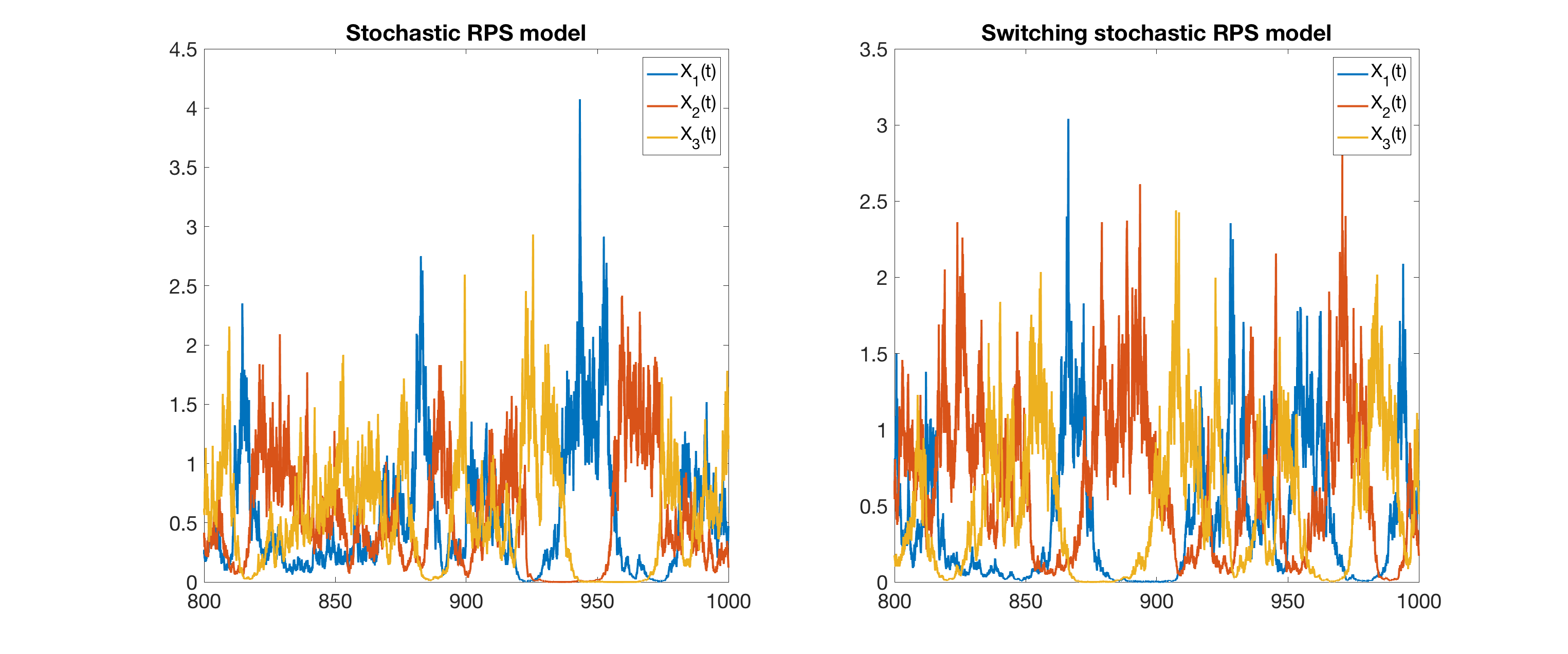}
 		\caption{Rock-paper-scissors model. Left: A sample
                  trajectory of the SDE model. Right: A sample
                  trajectory of the SSDE model. Note how in both settings the populations spend long times close to extinction. }
\label{fig12}\end{center}
\end{figure}
\begin{figure}[!htbp]
 	\begin{center}
 		\includegraphics[width = \linewidth]{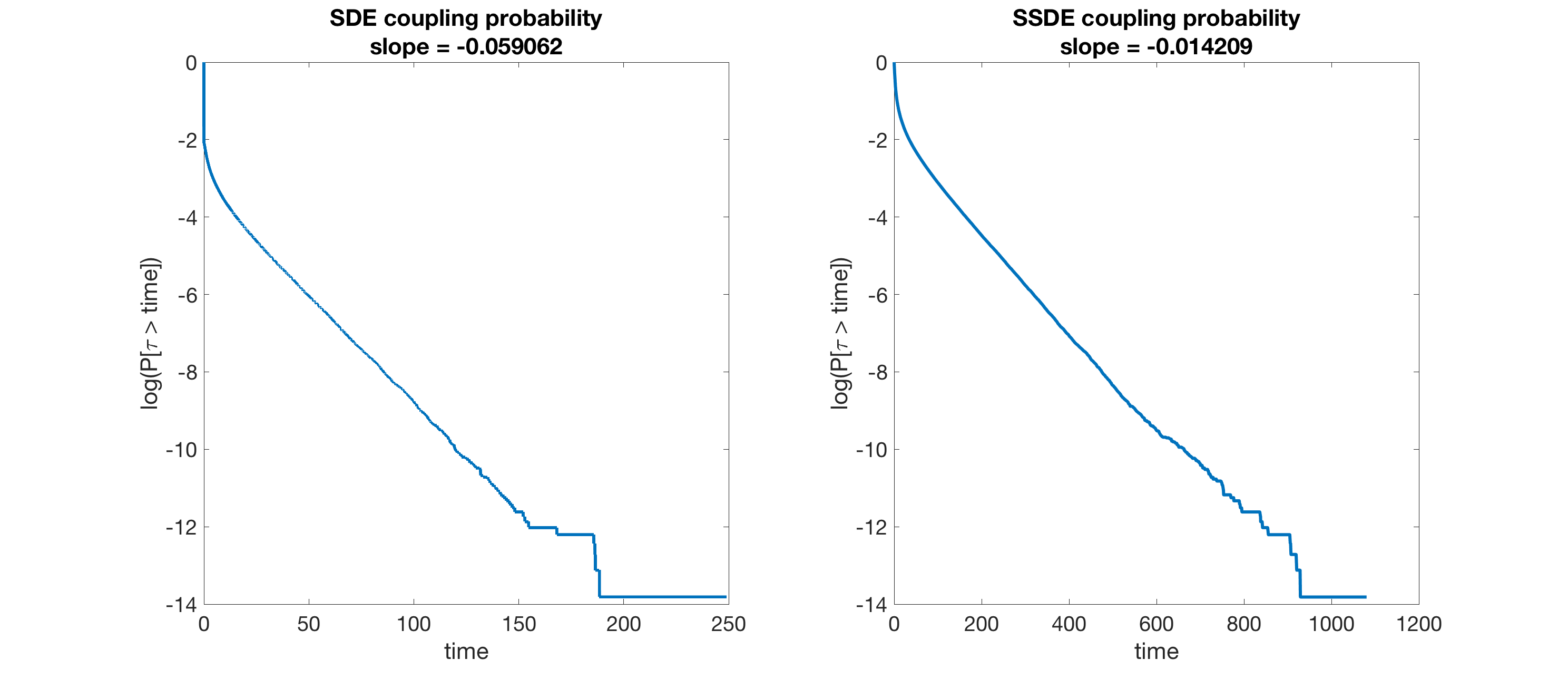}
 		\caption{Rock-paper-scissors model. One can see the exponential tails of the coupling time for the
                  SDE and SSDE models. }
\label{fig11}\end{center}
\end{figure}

\section{Discussion}\label{s:disc}

We have studied three different classes of models for stochastic population dynamics. The first one involves the usual framework of stochastic differential equations -- this setting has been used extensively in the literature \citep{SBA11,EHS15,HNY16, HN16, B18}. The second class of models we look at has at its core piecewise deterministic Markov processes \citep{D84}. Here the environment can switch randomly between a finite number of environmental states, and each environmental state has its own system of differential equations for the dynamics. These environmental switches can model seasonal or daily changes and can significantly change the long term dynamics \citep{BL16, HN20,HS19}. The last class of models we look at is that of stochastic differential equations with switching (SSDE), which combines white noise fluctuations in each fixed environment and random switches between a finite number of environments. The switching dynamics for PDMP and SSDE are governed by a process $r(t)$ which jumps from one state to another according to
\begin{equation}\label{e:trann}\begin{array}{ll}
&\disp \PP\{r(t+\Delta)=j~|~r(t)=i, \BX(s),r(s), s\leq t\}=q_{ij}(\BX(t))\Delta+o(\Delta) \text{ if } i\ne j \
\hbox{ and }\\
&\disp \PP\{r(t+\Delta)=i~|~r(t)=i, \BX(s),r(s), s\leq t\}=1+q_{ii}(\BX(t))\Delta+o(\Delta)
\end{array}
\end{equation}
where $q_{ij}(\bx)$ is the transition rate matrix. Our general framework allows eco-environmental feedback loops \citep{SL12,JLS94,CWH09,MLPS16}. The simplest case is when the switching rates are density independent, $q_{ij}(\bx)=q_{ij}$ and $r(t)$ becomes an independent Markov chain. In this special case, one has independent exponentially distributed waiting times in each environmental state.

\textbf{Persistence results.} We give powerful persistence results in all three settings. Our results show that a multidimensional version of the invasibility criterion \citep{C82, CE89, C94, C18, ESAH19} guarantees coexistence: If each subcommunity of species that persists, and is characterized by an invariant probability measure $\mu$, can be invaded by at least one
species that is not part of the subcommunity, then the full community of $n$ species persists. We compute the invasion rate of species $i$ into the subcommunity characterized by the invariant probability measure $\mu$ by averaging the per-capita growth rate of $i$ according to $\mu$. For SDE the averaging is done over the densities of the species. For PDMP and SSDE the averaging is over species densities as well as over all the environmental states.
The intuition behind this result is the following: if the process gets close to the extinction boundary, the species which are close to extinction will grow/decay exponentially fast according to their invasion rates. Since at least one invasion rate is positive, the process gets pushed away from the boundary and extinction is not possible.

If the ecosystem has only one species we show, in both the PDMP and SSDE settings, that coexistence is possible if the weighted average of the per capita growth rates at $0$ is positive. The weights are given by the stationary distribution $(\nu_1,\dots,\nu_{n_0})$ of the Markov chain $r(t)$ which governs the dynamics of the environment. This implies that coexistence is possible even if in some of the environmental states there is extinction (sinks) as long as, loosely speaking, more time is spent in the environments which exhibit coexistence (sources).

A key novelty is that we are able to characterize explicitly the `random equilibrium distribution' of a persistent ecological system. We accomplish this by developing rigorous numerical approximation methods for finding the invariant probability measure, or stationary distribution, of a stochastic process. This new method is significantly more accurate than previous approaches. The numerical methods combined with the analytic results give us a better understanding of what stochastic coexistence really looks like.

\textbf{Quasistationary distributions.} Our work highlights some important caveats to coexistence theory. Sometimes the coexistence criteria are satisfied and the species converge to a stationary distribution on $\R_+^{n,\circ}$. However, this stationary distribution might concentrate close to the extinction set $\partial \R_+^n$. This effect shows that, even though the theory implies persistence, the species might spend a long time at very low densities and go extinct due to demographic stochasticity. If this is the case for a specific model, one should not neglect demographic stochasticity.  If a population goes extinct almost surely in finite time, one can condition it on not going extinct at time $t$ and then seeing what this conditional distribution looks like as $t\to \infty$ \citep{SE04, MV12, KS12,CV16, HK19, HQSY20}. This limiting distribution is called the \textit{quasistationary distribution}. The invariant probability measures for both the Beddington-DeAngelis and the rock-paper-scissors models suggest that more realistic approaches should include the effects of demographic stochasticity and look into numerical and analytical approaches to finding quasistationary distributions.

We explore 3 different models in the SDE/PDMP/SSDE settings.

\textbf{1) Lotka--Volterra Competition.}

In his very influential paper \cite{H61} studied the competitive exclusion principle, which says in its simplest form that two species competing for one resource cannot coexist, and gave the following possible solution: Changing environmental conditions prevent species from reaching equilibrium conditions, and this can promote coexistence. Hutchinson's explanation was that if only few species would coexist at an equilibrium but we see many in nature, this would imply that there is no equilibrium. Fluctuating environmental conditions would favor different species at different times, so that no species can dominate. A key component of Hutchinson's explanation was that the environmental fluctuations need to be on the correct time scale: slow or fast fluctuations would not work. We show through an example in the SSDE setting that fast fluctuations between two environments in which species $1$ goes extinct, can rescue this species from extinction and facilitate coexistence.

Our numerical analysis shows in the SDE setting that the white noise environmental fluctuations turn the stable equilibrium of the two species into stationary distributions that concentrate around the deterministic equilibria. How spread the invariant measure is seems to be related to the intraspecific competition rates. We find that the higher the intraspecific competition rate is for species $i$, the smaller the spread of the stationary distribution is in direction $i$ (Figure \ref{fig1LV}). High intraspecific competition makes it harder for the species to leave the region that is close to the deterministic stable equilibria.

In the PDMP setting we find the following behavior. If the switching is slow the species spend a long time in each environment. Because of this they have time to get close the equilibria. This is in line with the explanation for the paradox of the plankton by \cite{H61}. As switching becomes faster, the dynamics spends less time in a fixed environment, and the species do not have time to get close to the stable equilibrium from that environment. Our numerical simulations (see Figure \ref{fig2LV} Middle) reflect this fact. As we increase the switching rate and we go from intermediate switching to fast switching, the stationary distribution once again becomes more singular. This is because, as explained above, for very fast switching the dynamics approaches the one given by a mixed deterministic ODE and therefore the mass of the stationary distribution concentrates around the equilibrium of this mixed ODE.

\textbf{2) Beddington--DeAngelis predator-prey dynamics.}
In the deterministic setting, if there is no intraspecific competition among predators, there are situations when the dynamics has a limit cycle. If one adds white noise fluctuations  one gets interesting behaviour. The level sets of the invariant probability density look qualitatively like concentric closed trajectories. The density of the invariant probability measure is highest close to the $y$ axis. This shows that in this system, the prey can become low and stay low for a long time, while the predator has a significant density and takes a long time to die out in the absence of a food source. When the predator population finally decreases, the prey increases again, which causes the predator to increase after some time . This cycle would then get repeated. Our results highlight that even though the theory implies coexistence, a more realistic system might go extinct. If one spends a long time close to the boundary this demographic stochasticity might induce extinction.

In the PDMP model, switching between the two environments creates a coexistence situation where the invariant measure is qualitatively similar to the occupation measure of a limit cycle. However, in this setting the support of the stationary distribution seems to be bounded away from the extinction set -- there is no concentration near the extinction set. As the switching speed increases the invariant probability measures from the two environments become closer and closer. The dynamics with fast switching will have a unique limit cycle that is significantly larger than the limit cycles from the two deterministic environments. This is another example of how environmental fluctuations significantly change the dynamics by making the species densities oscillate at greater amplitudes (more than double the amplitudes from each fixed environment).

\textbf{3) Rock-paper-scissors dynamics.}
In the SDE setting the stationary distribution looks like a noisy heteroclinic cycle. The stationary distribution is supported on a compact set and has high densities close to the boundary. As we increase the density of the third species $z$, the support of the stationary distribution decreases -- as $z$ goes from $0.25$ to $2$ the support of the stationary distribution decreases tenfold.

If the dynamics switches between persistent and extinction prone rock-paper-scissors systems one can still get overall persistence. Even though the switching is random, we expect the same phenomenon to arise when the switching is periodic. This would describe seasonality and annual cycles. In various regions which exhibit strong seasonal variation, organisms have developed survival mechanisms like hibernation and migration in order to survive seasons with resource shortages. Our results show that as long as the bad environment is not too harsh, the species can still persist.

\textbf{Acknowledgements:} The authors acknowledge support from the NSF through the grants DMS-1853463 for Alexandru Hening and DMS-1813246 for Yao Li.

\bibliographystyle{agsm}
\bibliography{LV}

\appendix
\section{Technical assumptions}
\subsection{SDE}\label{s:a_SDE} We need to make some technical assumptions for Theorem \ref{t:p_sde} to hold.
\begin{asp}\label{a.nonde}  The coefficients of \eqref{e:system} satisfy the following conditions:
\begin{itemize}
\item[(1)] $\diag(g_1(\bx),\dots,g_n(\bx))\Gamma^\top\Gamma\diag(g_1(\bx),\dots,g_n(\bx))=(g_i(\bx)g_j(\bx)\sigma_{ij})_{n\times n}$
is a positive definite matrix for any $\bx\in\R^{n}_+$.
\item[(2)] $f_i(\cdot), g_i(\cdot):\R^n_+\to\R$ are locally Lipschitz functions for any $i=1,\dots,n.$
\item[(3)] There exist $\bc=(c_1,\dots,c_n)\in\R^{n,\circ}_+$ and $\gamma_b>0$ such that
\begin{equation}\label{a.tight2}
\limsup\limits_{\|x\|\to\infty}\left[\dfrac{\sum_i c_ix_if_i(\bx)}{1+\bc^\top\bx}-\dfrac12\dfrac{\sum_{i,j} \sigma_{ij}c_ic_jx_ix_jg_i(\bx)g_j(\bx)}{(1+\bc^\top\bx)^2}+\gamma_b\left(1+\sum_{i} (|f_i(\bx)|+g_i^2(\bx))\right)\right]<0.
\end{equation}
\end{itemize}
\end{asp}

Parts (2) and (3) of Assumption \ref{a.nonde} guarantee the existence and uniqueness of strong solutions to \eqref{e:system}.
Part (1) ensures that the solution to \eqref{e:system} is a non degenerate diffusion - this means that the noise is truly $n$ dimensional.
Condition (3) also implies the tightness of the family of transition probabilities of the solution to \eqref{e:system}. Most common ecological models satisfy condition \eqref{a.tight2}. This condition is a requirement that there is a strong drift towards zero when the population size is large. This usually holds if the intraspecific competition is strong enough.
\begin{itemize}
\item In the logistic model
$$dX(t)=X(t)[a-bX(t)]dt+\sigma X(t)dB(t), b>0$$ the condition \eqref{a.tight2} is satisfied for any $c>0$.

\item In the competitive Lotka-Volterra model
$$dX_i(t)=X_i(t)\left[a_i-\sum_{j}b_{ji}X_j(t)\right]dt+X_i(t)g_i(\BX(t))dE_i(t),$$
with $b_{ji}> 0, j,i=1,\dots,n$,  if $$\sum_{i=1}^n g_i^2(\bx)<K(1+\|\bx\|+\bigwedge_{i=1}^n g_i^2(\bx))$$ then
 \eqref{a.tight2} is satisfied with $\bc=(1,\dots,1)$.
\item In the predator-prey Lotka-Volterra model
$$
\begin{cases}
dX(t)=X(t)[a_1-b_1X(t)-d_1Y(t)]dt+X(t)g_1(X(t), Y(t))dE_1(t)\\
dY(t)=Y(t)[-a_2-b_2Y(t)+d_2X(t)]dt+Y(t)g_2(X(t), Y(t))dE_2(t),\\
\end{cases}
$$
with $b_1,b_2>0, d_1,d_2\geq 0, a_2\geq 0$, if $$\sum_{i=1}^2 g_i^2(x,y)<K(1+x+y+g_1^2(x,y)\wedge g_2^2(x,y)),$$
then one can show that
 \eqref{a.tight2} is satisfied with $\bc=(d_2,d_1)$.
\end{itemize}
\subsection{PDMP}\label{s:a_PDMP} It is well-known that a process $(\BX(t),r(t))$ satisfying \eqref{e1-pdm} and \eqref{e:tran}
is a Markov process with generator $\Lom$ acting on functions $G:\R_+^n\times\CN\mapsto\R$ that are continuously differentiable in $\bx$ for each $k\in\CN$ as
\begin{equation}
\Lom G(\bx, k)=\sum_{i=1}^n x_if_i(\bx,k)\frac{\partial G}{\partial x_i}(\bx,k)+\sum_{l\in\CN}q_{kl}(\bx)G(\bx,l).
\end{equation}

In order to have a well behaved process $(\BX(t), r(t))$ we make the following assumptions:
\begin{itemize}
  \item  There is a vector $\bc\in\R^{n,\circ}_+$ such that $\Lom (1+\bc^\top\bx)=\sum_ic_ix_if(x_i,k)<0$. Then $\BX(t)$ eventually enters a compact set and never leaves it. This makes it possible to reduce the dynamics to compact sets. Most models will satisfy this assumption. If it is violated one needs a different assumption which implies that  $\BX(t)$ returns fast to compact subsets of $\R_+^{n,\circ}$.
  \item Let $\gamma^+(\bx)$ denote the orbit set, that is, the set of all points from $\R_+^{n,\circ}$ that are reachable by some possible trajectory of \eqref{e1-pdm} with $\BX(0)=\bx$. Let $$\Gamma=\bigcap_{\bx\in\R^{n,\circ}_+}\bar{\gamma^+(\bx)}$$
      be the set of all points that are weakly accessible i.e. lie in the intersection of all the closures of orbit sets. We assume that $\Gamma\neq \emptyset$.  This implies that there is at least one point which can be close to trajectories started from any initial point from $\R_+^{n,\circ}$.
  \item There exists a point $\bx_0\in\Gamma$ which satisfies the strong bracket condition (see \cite[Definition 4.3]{BBMZ15} as well as \cite{B18, BHS18}). This condition ensures that the process is not too degenerate and is, in a sense, well behaved.
\end{itemize}
\subsection{SSDE} Many of the well-known facts about SDE carry over to SSDE. However, the proofs can be quite technical when the generator $q(\bx)$ of the switching process depends on the density $\bx$ of the species. We refer the reader to \cite{YZ09, NYZ17} for conditions and proofs of the existence and uniqueness of solutions, the Feller property, recurrence, transience, and ergodicity. In our models, as long as all the involved coefficients are smooth enough and $q(\bx)$ is bounded, continuous and irreducible the only additional ingredients are:
\begin{enumerate}
  \item The process returns quickly to compact subsets of the state space. This requires a boundedness/tightness assumption like \eqref{a.tight2} or the assumptions of Proposition 4.13 from \cite{B18}.
  \item Just like for PDMP, in order to get the existence of a stationary distribution, we need conditions which ensure some type of irreducibility and non-degeneracy - see the discussion on PDMP and \cite{B18}.
\end{enumerate}

\section{The numerical scheme}\label{s:a_conv}

\subsection{Convergence of the algorithm}
Let us prove the convergence of our algorithm. Let $u^{*}$ be
the density function of the invariant probability measure on $(0,
\infty)^{\infty} \times \mathcal{N}$. We let $\mathcal{N} = \{1\}$ for the
case of SDE. Without loss of generality,
assume the domain $\mathcal{D} = [a_{1}, b_{1}] \times \cdots \times
[a_{n} b_{n}]$ is divided into $N^{n}$ bins, with a step size $h =
(b_{1} - a_{1})/N = \cdots = (b_{n} - a_{n})/N$. Let $\mathbf{u}^{*}
\in \mathbb{R}^{N^{n}n_{0}}$ be the discrete probability density
function of $u^{*}$, such that $\mathbf{u}^{*}_{i_{1}, \cdots, i_{n},
  k} = u(x_{i_{1}}, \cdots, x_{i_{n}}, k)$, where $(x_{i_{1}}, \cdots,
x_{i_{n}})$ is the coordinate of the center of the $(i_{1}, \cdots,
i_{n})$-th box in the grid, and $1 \leq k \leq n_{0}$ is a discrete
state. Since a numerical solver of SDE, SSDE, or PDMP has some error,
the invariant probability measure of the Monte Carlo sampler is
usually different from $\mu^{*}$. We denote the invariant probability
measure of the Monte Carlo sampler by $\hat{\mu}$, and define the
corresponding probability density function $\hat{u}$ and discrete
probability distribution $\mathbf{\hat{u}}$ in a similar
way.

The idea of proof is straightforward. The main source of the error terms is $\mathbf{v} -
\mathbf{\hat{u}}$, which is ``noisy'' enough to be treated as a
centralized random vector. The optimization problem projects the error
term $\mathbf{e} = \mathbf{v} - \mathbf{\hat{u}}$ to a lower dimensional space. If
in addition $\mathbf{\hat{u}}$ is sufficiently close
to $\mathbf{u}^{*}$, then we can control the error of ${\bm \ell}$.

Similar as in  \cite{dobson2019efficient}, we need the following assumptions to prove
the convergence result.

\begin{itemize}
  \item[(a)] Entries of $\mathbf{e} = \{ e_{i_{1}, \cdots,
      i_{n},k}\}_{1 \leq i_{1}, \cdots, i_{n} \leq N, 1 \leq k \leq
      n_{0}}$ are uncorrelated random variables with expectation $0$
    and variance no greater than $\zeta^{2}$.
\item[(b)] The boundary value problem for equation \eqref{e:FPS} (or
  equation \eqref{switchFPE} for SSDE) is well-posed and has a unique
  solution.
\item[(c)] The finite difference scheme for equation \eqref{e:FPS} (or
  equation \eqref{switchFPE} for SSDE) is convergent for a boundary
  value problem on $\mathcal{D} \times \mathcal{N}$ with $L^{\infty}$
  error $O(h^{p})$.
\end{itemize}

The $L^{2}$ error $\|\ell_{*} - u^{*}\|$ is measured by $h^{n/2}
\mathbb{E}[\| {\bm \ell} - \mathbf{u}^{*}\|]$, which is the numerical
integration of the error term over the grid. We have the following theorem.

\begin{thm}
\label{convergence}
Assume (a)-(c) holds. We have
$$
  h^{n/2} \mathbb{E}[\| {\bm \ell} - \mathbf{u}^{*} \|_{2}] \leq
  O(h^{1/2} \zeta) + O(h^{p}) + \| \hat{u} - u^{*}\|_{2} \,.
$$
\end{thm}
\begin{proof}
The proof is generalized from Theorem 2.1 of \cite{dobson2019efficient}. Consider
equation \eqref{e:FPS} (resp. \eqref{switchFPE}) on the extended
domain $\tilde{\mathcal{D}} = [a_{1} - h, b_{1} + h] \times \cdots \times [a_{n}
-h, b_{n} + h] \times \mathcal{N}$ with boundary condition
\begin{equation}
\label{bvp}
  \mathcal{L} w = 0 \quad \mbox{ in } \tilde{\mathcal{D}}, \quad w =
  u^{*} \quad \mbox{ on } \partial \tilde{\mathcal{D}}
\end{equation}
By assumption (b), $u^{*}|_{\tilde{\mathcal{D}}}$ solves equation
\eqref{bvp}.

Solving this boundary value problem by a finite difference method, we
have a linear system
$$
\begin{bmatrix}
\mathbf{A}& \mathbf{0}\\ \mathbf{B} & \mathbf{C}
\end{bmatrix}
\begin{bmatrix}
\mathbf{u}^{lin}\\ \mathbf{u}_{0}
\end{bmatrix}
=
\begin{bmatrix}
\mathbf{0} \\ \mathbf{0}
\end{bmatrix} \,,
$$
where the matrix $\mathbf{A}$ is same as the linear constraint in the
optimization problem \eqref{optimization}, and the matrices $\mathbf{B}$ and
$\mathbf{C}$ come from the boundary condition. Therefore,
$\mathbf{u}^{lin}$ satisfies the linear constraint $\mathbf{A}
\mathbf{u} = \mathbf{0}$. In addition, by assumption (c), we have
$$
  \| \mathbf{u}^{lin} - \mathbf{u}^{*} \|_{\infty} = O(h^{p}) \,.
$$

Let $\mathbf{u}$ be the solution to the optimization problem
\eqref{optimization}. Let $P$ be the projection matrix to
$\mathrm{Ker}( \mathbf{A})$. Since $\mathbf{u} \in \mathrm{Ker}(
\mathbf{A})$, we have
$$
  \mathbf{u} - \mathbf{u}^{lin} = P \mathbf{v} - \mathbf{u}^{lin} - P(
  \mathbf{v} - \mathbf{u}^{lin}) = P( \mathbf{v} - \mathbf{\hat{u}}) +
  P( \mathbf{\hat{u}} - \mathbf{u}^{lin}) \,.
$$
This implies
$$
  \mathbb{E}[\| \mathbf{u} - \mathbf{u}^{lin}\|_{2} ]\leq  \mathbb{E}[\|
  P(\mathbf{v} - \mathbf{\hat{u}})\|_{2}] +  \| P( \mathbf{\hat{u}} -
  \mathbf{u}^{lin})\|_{2} \,.
$$
The second term has bound
\begin{align*}
  \| P( \mathbf{\hat{u}} -
  \mathbf{u}^{lin})\|  &\leq \| \mathbf{\hat{u}} -
  \mathbf{u}^{lin}\|_{2} \leq \| \mathbf{\hat{u}} -
  \mathbf{u}^{*}\|_{2} + \| \mathbf{u}^{*} - \mathbf{u}^{lin}\|_{2} \\
&\leq O(h^{p})O(N^{n/2}) + \| \hat{u} - u^{*} \|_{2} O(N^{n/2})\,.
\end{align*}
By assumption (a), $\mathbf{e} = \mathbf{v} - \mathbf{\hat{u}}$ is a
random vector with uncorrelated entries. Note that the dimension of
$\mathrm{Ker}(\mathbf{A})$ is $d(N) := N^{n} - (N-2)^{n}
=O(N^{n-1})$. Let $S $ be an orthogonal matrix such that the first $d(N)$ columns of $S^{T}$
form an orthonormal basis of $\mathrm{Ker}(\mathbf{A})$. Then $S$ is a
change-of-coordinate matrix that changes $\mathrm{Ker}( \mathbf{A})$
to the subspace generated by the first $d(N)$ coordinates.

Let $S^{T} = [ \mathbf{s}_{1}, \cdots, \mathbf{s}_{N^{n}}]$ and $S \mathbf{e} = [\hat{e}_{1}, \cdots, \hat{e}_{N^{n}}]^{T}$. The
projection gives
$$
  P \mathbf{e} = \sum_{i = 1}^{d(N)} \hat{e}_{i} \mathbf{s}_{i}
$$
This implies
$$
  \mathbb{E}\left[ \| P \mathbf{e} \|_{2}\right] = \mathbb{E}\left[ \left ( \sum_{i =
      1}^{d(N)} \hat{e}_{i}^{2} \right )^{1/2}\right] \leq \left (\mathbb{E}\left[  \sum_{i =
      1}^{d(N)} \hat{e}_{i}^{2} \right] \right )^{1/2}
$$
because $S$ is an orthonormal matrix. By assumption (a), the entries of $\mathbf{e}$
are uncorrelated and have expectations $0$. This implies
$$
  \mathbb{E}[ \hat{e}_{i}^{2}] = \sum_{i = 1}^{N^{d}}S_{ji}^{2}
  \mathbb{E}[ e_{j}^{2}] \leq \zeta^{2}
$$
because $S$ is an orthogonal matrix. This gives
$$
  \mathbb{E}[ \| \mathbf{v} - \mathbf{\hat{u}}\|_{2}] \leq \sqrt{d(N)
  } \zeta \,.
$$

Finally, the triangle inequality yields
$$
  \mathbb{E}[\| \mathbf{u} - \mathbf{u}^{*} \|_{2} ]\leq \mathbb{E}[\|
  \mathbf{u} - \mathbf{u}^{lin} \|_{2} ] + \| \mathbf{u}^{lin} -
  \mathbf{u}^{*}\|_{2} \leq \sqrt{d(N)}\zeta + O(N^{n/2}h^{p}) + \|
  \hat{u} - u^{*}\|_{2} O(N^{n/2}) \,.
$$
Since $h = O(N^{-1})$ and $d(N) = O(N^{n-1})$, the proof is completed
by multiplying both sides by $h^{n/2}$.

\end{proof}

We remark that the empirical accuracy of $\mathbf{u}$ is much better
than the theoretical bound given in Theorem \ref{convergence}. This is
because the projected error term $P\mathbf{e}$ tends to concentrate at the
boundary of the domain. This can be justified by computing principal
angles between $\mathrm{Ker}( \mathbf{A})$ and the subspace spanned by
the boundary of the domain. See our discussion in \cite{dobson2019efficient}.

\subsection{Sensitivity analysis}

It follows from Theorem \ref{convergence} that the sample quality
plays a key role in the accuracy of the data-driven solver. If the
difference between $\hat{\mu}$ and $\mu^{*}$ is too large, the
accuracy of the numerical solution ${\bm \ell}$ could be
questionable. In other words, it is crucial to know the sensitivity of
the invariant probability measure against the perturbation caused by
the time-discretization when simulating an SDE (or SSDE) numerically.

Let $d$ be a distance between probability measures. Let $P^{t}$ and
$\hat{P}^{t}$ be the transition probability kernel of the SDE (or
SSDE) and its numerical scheme, respectively.  It follows from
\cite{dobson2019using} that the sensitivity of $\mu^{*}$ can be controlled by the following triangle inequality
$$
  d(\mu^{*}, \hat{\mu}) \leq d(\mu^{*} P^{T} ,\mu^{*} \hat{P}^{T} ) + d(\mu^{*}
  \hat{P}^{T} , \hat{\mu} \hat{P}^{T}) \,,
$$
where $T > 0$ is a fixed finite time span. If $P^{T}$ is contracting in the metric space $(\mathcal{P}, d)$,
where $\mathcal{P}$ is the collection of all probability measures on
$(0, \infty)^{n} \times \mathcal{N}$ then $$d(\mu^{*}\hat{P}^{T} , \hat{\mu} \hat{P}^{T}) \leq \alpha d(\mu^{*},
  \hat{\mu})$$ for some $\alpha < 1$. If in addition $d(\mu^{*}
  P^{T} ,\mu^{*} \hat{P}^{T} ) \leq \epsilon$ for some $\epsilon > 1$, we have $d(\mu^{*}, \hat{\mu})
  \leq \epsilon (1 - \alpha)^{-1}$.

In \cite{dobson2019using}, we let $d = d_{w}$, where $d_{w}$ is the
1-Wasserstein distance induced by the metric $\rho(x,y) = \min\{1, |x
- y| \}$. Then we use linear extrapolation of the truncation error
to estimate $d_{w}(\mu^{*}P^{T} ,\mu^{*} \hat{P}^{T} ) $, and use
the coupling method to compute the rate of contraction $\alpha$.

Let $dt$ be
the time step size of the numerical SDE (or SSDE) solver. Let
$\hat{X}^{dt}_{n}$ be the state of $n$-th step of the numerical
trajectory with a time step $dt$. Let $T = Kdt$ for some even integer $K$. The
estimator of $d_{w}(\mu^{*}P^{T} ,\mu^{*} \hat{P}^{T} ) $ is given by
$$
  d_{w}(\mu^{*}P^{T} ,\mu^{*} \hat{P}^{T} )  \approx
  \frac{c}{N}\sum_{i = 0}^{N - 1}\rho(\hat{X}^{dt}_{i,T},
  \hat{X}^{2dt}_{i,T})  \,,
$$
where $\hat{X}^{dt}_{i, T} = \hat{X}^{dt}_{K(i+1)}$, and
$\hat{X}^{2h}_{i, T}$ is the terminal value of a ``parallel'' time-$2dt$ trajectory
with length $T$. In other words, after each time $T$, we reset
$\hat{X}^{2h}_{0} = \hat{X}^{h}_{iT}$ and compute $\hat{X}^{2h}_{n}$
using the same random term as in $\hat{X}^{h}_{n}$, and denote $\hat{X}^{2h}_{i,T} =
\hat{X}^{2h}_{K/2}$ after $K/2$ steps. The parameter $c$ comes from the
linear extrapolation. For example, $c = 1$ if the numerical scheme has
strong convergence with order $1$.

The idea is that $\hat{X}^{h}_{iT}$ for $i = 0, \cdots, N$ are sampled from
$\hat{\mu}$. Starting from each $\hat{X}^{h}_{iT}$, $\rho(X_{T},
\hat{X}^{dt}_{K})$, the distance between true SDE trajectory and its
numerical approximation, can be obtained from a linear extrapolation of $\rho(\hat{X}^{dt}_{i,T},
  \hat{X}^{2dt}_{i,T})$. A calculation shows that
$$
  \frac{c}{N}\sum_{i = 0}^{N - 1}\rho(\hat{X}^{dt}_{i,T},
  \hat{X}^{2dt}_{i,T})
$$
gives an approximate upper bound of $d_{w}(\mu^{*}P^{T} ,\mu^{*}
\hat{P}^{T} )$. See our discussion in \cite{dobson2019using}.

The contraction rate $d(\mu^{*}\hat{P}^{T} , \hat{\mu} \hat{P}^{T})$
can be approximated by using the coupling method. It is well known
that $( \mathcal{X}_{t}, \mathcal{Y}_{t}) \in \mathbb{R}^{2n}$ is a
coupling of $X_{t} \in \mathbb{R}^{n}$ if two marginal distributions
have the same law of $X_{t}$.  Let $( \mathcal{X}_{t}, \mathcal{Y}_{t})$ be a coupling such
    that if $\mathcal{X}_{t} = \mathcal{Y}_{t}$, then $\mathcal{X}_{s}
    = \mathcal{Y}_{s}$ for all $s > t$. Let $\tau_{c} = \inf_{t} \{ \mathcal{X}_{t} = \mathcal{Y}_{t} \}$ be
  the coupling time. Then it follows from the coupling inequality that
$$
  \mathrm{d}_{w}( \mu \hat{P}^{t}, \nu \hat{P}^{t}) \leq \mathbb{P}_{\mu, \nu}[ \tau_{c} > t] \,.
$$

It is usually very difficult to sharply estimate $\mathbb{P}[ \tau_{c}
> t]$ using rigorous methods. Instead, we can run Monte Carlo simulation
to obtain the exponential tail of $\mathbb{P}[ \tau_{c} > t]$
numerically. If numerical simulation gives $\mathbb{P}[ \tau_{c} > t]
\approx e^{-c t}$, then the rate of contraction at time $T$ is
approximately $\alpha = e^{-cT}$. This gives an estimator
$$
  d_{w}(\mu^{*},\hat{\mu})  \approx
  \frac{c}{N(1 - e^{-cT})}\sum_{i = 0}^{N - 1}\rho(\hat{X}^{dt}_{i,T},
  \hat{X}^{2dt}_{i,T})  \,.
$$

It remains to comment regarding the suitable construction of the coupling. The
simplest coupling is to run to independent numerical trajectories of
$\hat{X}^{dt}_{n}$ until they meet. This coupling has very low
efficiency. Also it is hard to have trajectories meet in
a continuous space. Instead, we can run two numerical trajectories
using either the same noise terms (synchronous coupling), or reflected
noise terms (reflection coupling), until they are close enough. Then
we can compare the probability density of the next step of the
numerical trajectory. Assume the probability density functions of
$\mathcal{X}_{t+1}$ and $\mathcal{Y}_{t+1}$ are $f_{X}$ and $f_{Y}$
respectively. Then with probability
$$
  \int_{\mathbb{R}^{n}} \min \{ f_{X}(\mathbf{x}), f_{Y}( \mathbf{x})
  \} \mathrm{d} \mathbf{x} \,,
$$
we can couple $\mathcal{X}_{t+1}$ and $\mathcal{Y}_{t+1}$ at step
$t+1$. We refer to work by \cite{dobson2019using, li2020numerical} for the full details.

\end{document}